\documentclass[a4paper,onecolumn,accepted=2020-03-27]{quantumarticle}
\pdfoutput=1

\input{restricted.sty}

\title{Number-Theoretic Characterizations of Some Restricted
  Clifford+T Circuits}

\author{Matthew Amy}
\affiliation{Department of Mathematics and Statistics, Dalhousie
  University, Halifax, NS, Canada}
\orcid{0000-0003-3514-420X}

\author{Andrew N.\ Glaudell}
\affiliation{Institute for Advanced Computer Studies and Joint
  Center for Quantum Information and Computer Science, \\
  University of Maryland, College Park, MD, USA}
\affiliation{Joint Quantum Institute, University of Maryland,
  College Park, MD, USA}

\author{Neil J.\ Ross}
\affiliation{Department of Mathematics and Statistics, Dalhousie
  University, Halifax, NS, Canada}
\orcid{0000-0003-0941-4333}

\begin{document}

\maketitle

\begin{abstract}
  Kliuchnikov, Maslov, and Mosca proved in 2012 that a $2\times 2$
  unitary matrix $V$ can be exactly represented by a single-qubit
  Clifford+$T$ circuit if and only if the entries of $V$ belong to the
  ring $\Z[1/\sqrt{2},i]$. Later that year, Giles and Selinger showed
  that the same restriction applies to matrices that can be exactly
  represented by a multi-qubit Clifford+$T$ circuit. These
  number-theoretic characterizations shed new light upon the structure
  of Clifford+$T$ circuits and led to remarkable developments in the
  field of quantum compiling. In the present paper, we provide
  number-theoretic characterizations for certain restricted
  Clifford+$T$ circuits by considering unitary matrices over subrings
  of $\Z[1/\sqrt{2},i]$. We focus on the subrings $\Z[1/2]$,
  $\Z[1/\sqrt{2}]$, $\Z[1/i\sqrt{2}]$, and $\Z[1/2,i]$, and we
  prove that unitary matrices with entries in these rings correspond
  to circuits over well-known universal gate sets. In each case, the
  desired gate set is obtained by extending the set of classical
  reversible gates $\s{X, CX, CCX}$ with an analogue of the Hadamard
  gate and an optional phase gate.
\end{abstract}

\section{Introduction}
\label{sec:intro}

Kliuchnikov, Maslov, and Mosca showed in \cite{KMM-exact} that a
2-dimensional unitary matrix $V$ can be exactly represented by a
single-qubit Clifford+$T$ circuit if and only if the entries of $V$
belong to the ring $\Z[1/\sqrt{2},i]$. This result gives a
number-theoretic characterization of single-qubit Clifford+$T$
circuits. In \cite{GS13}, Giles and Selinger extended the
characterization of Kliuchnikov et al. to multi-qubit Clifford+$T$
circuits by proving that a $2^n$-dimensional unitary matrix can be
exactly represented by an $n$-qubit Clifford+$T$ circuit if and only
if its entries belong to $\Z[1/\sqrt{2},i]$.

These number-theoretic characterizations provide great insight into
the structure of Clifford+$T$ circuits. As a result, single-qubit
Clifford+$T$ circuits are now very well understood
\cite{fallback,ma-remarks,kmm-approx,MA08,RS16}, 
and some of these results have even been extended to single-qubit 
circuits beyond the Clifford+$T$ gate set
\cite{bgs13, fgkm15, kbs14, ky15, ps18, r15}. In contrast, our
understanding of multi-qubit Clifford+$T$ circuits remains more
limited, despite interesting results
\cite{BS2015,devos,GKMR,Gr2014,BSW2016}. One of the reasons for this
limitation is that large unitary matrices over $\Z[1/\sqrt{2},i]$ are
hard to analyze. In order to circumvent the difficulties associated
with multi-qubit Clifford+$T$ circuits, restricted gate sets have been
considered in the literature. This led to important developments in
the study of multi-qubit Clifford, CNOT+$T$, and CNOT-dihedral
circuits \cite{ACR17,AMM,AMMR,AM,CH2018,MSdM2018,Sel}. Unfortunately,
the simpler structure of these restricted gate sets comes at a cost:
they are not universal for quantum computing.

In the present paper, our goal is to address both of these limitations
by considering restrictions of the Clifford+$T$ gate set which are
nevertheless universal for quantum computing. To this end, we study
circuits corresponding to unitary matrices over proper subrings of
$\Z[1/\sqrt{2},i]$, focusing on $\Z[1/2]$, $\Z[1/\sqrt{2}]$,
$\Z[1/i\sqrt{2}]$, and $\Z[1/2,i]$. For each subring, we find a set of
quantum gates $G$ with the property that circuits over $G$ correspond
to unitary matrices over the given ring. Writing $U_{2^n}(R)$ for the
group of $2^n \times 2^n$ unitary matrices over a ring $R$, our main
results can then be summarized in the following theorem.

\

\begin{theorem*}
  A $2^n\times 2^n$ unitary matrix $V$ can be exactly represented by
  an $n$-qubit circuit over
  \begin{enumerate}[label=(\roman*)]
    \item $\s{X, CX, CCX, H\otimes H}$ if and only if $V\in
      U_{2^n}(\Z[1/2])$,
    \item $\s{X, CX, CCX, H, CH}$ if and only if $V\in
      U_{2^n}(\Z[1/\sqrt{2}])$,
    \item $\s{X, CX, CCX, F}$ if and only if $V\in
      U_{2^n}(\Z[1/i\sqrt{2}])$, and
    \item $\s{X, CX, CCX, \omega H, S}$ if and only if $V\in
      U_{2^n}(\Z[1/2,i])$,
  \end{enumerate}
  where $\omega=e^{i\pi/4}$ and $F\propto\sqrt{H}$. Moreover, in
  (i)-(iv), a single ancilla is sufficient.
\end{theorem*}

The gate sets in items $(i)$--$(iv)$ of the above theorem are all
universal for quantum computing \cite{Aharonov03asimple,Shi2003}, and
we sometimes refer to circuits over these gate sets as
\emph{integral}, \emph{real}, \emph{imaginary}, and \emph{Gaussian}
Clifford+$T$ circuits, respectively. As a corollary to the above
theorem, we also obtain two additional characterizations of universal
gate sets.

\begin{corollary*}
  A $2^n\times 2^n$ unitary matrix $V$ can be exactly represented by
  an $n$-qubit circuit over
  \begin{enumerate}[label=(\roman*)]
    \item $\s{X, CX, CCX, H}$ if and only if $V=W/\sqrt{2}^q$ for some
      matrix $W$ over $\Z$ and some $q\in\N$, and
    \item $\s{X, CX, CCX, H, S}$ if and only if $V=W/\sqrt{2}^q$ for
      some matrix $W$ over $\Zi$ and some $q\in\N$.
  \end{enumerate}
  Moreover, in (i) and (ii), a single ancilla is sufficient.
\end{corollary*}

As a final corollary to the theorem above, we refine the
characterizations $(iii)$ and $(iv)$ by showing that in these cases a
matrix can be represented by an ancilla-free circuit if and only if it
has determinant~1.

\begin{corollary*}
  Let $n\geq 4$. A $2^n\times 2^n$ unitary matrix $V$ can be exactly
  represented by an $n$-qubit ancilla-free circuit over
  \begin{enumerate}[label=(\roman*)]
    \item $\s{X, CX, CCX, F}$ if and only if $V\in
      U_{2^n}(\Z[1/i\sqrt{2}])$ and $\det V=1$, and
    \item $\s{X, CX, CCX, \omega H, S}$ if and only if $V\in
      U_{2^n}(\Z[1/2,i])$ and $\det V=1$.
  \end{enumerate}
  In (i) and (ii), the requirement that $\det V=1$ can be dropped for
  $n<4$.
\end{corollary*}

The characterization of ancilla-free real and integral Clifford+$T$
circuits remains an open question but we conjecture that they
correspond to a strict subgroup of the group of unitaries with
determinant~1.

Restrictions similar to the ones considered here were previously
studied in the context of foundations \cite{RL2002}, randomized
benchmarking \cite{HFGW18}, and graphical languages for quantum
computing \cite{BK2019,JPV2017,vilmart2019}. Furthermore, our study
fits within a larger program, initiated by Aaronson and others, which
aims at classifying quantum operations. Such classifications exist for
classical reversible operations \cite{AGS17}, for stabilizer
operations \cite{GS18}, and for beam-splitter interactions \cite{ab},
but no classification is known for a universal family of quantum
operations suited for fault-tolerant quantum computing. In this
context, our work can be seen as a partial classification of the
universal extensions of the set of classical reversible gates $\s{X,
  CX, CCX}$. This perspective is illustrated in
\Cref{fig:inclusionlattice}, which depicts a fragment of the lattice
of subgroups of $U_n(\Z[1/\sqrt{2},i])$ where, for conciseness, we
wrote $\D$ for the ring $\Z[1/2]$ so that the rings $\Z[1/\sqrt{2}]$,
$\Z[1/i\sqrt{2}]$, $\Z[1/2,i]$ and $\Z[1/\sqrt{2},i]$ are denoted by
$\Drtwo$, $\Drminustwo$, $\Di$, and $\Domega$, respectively.

\begin{figure}[t]
  \begin{center}
    \begin{tikzpicture}[scale=0.8]
      \draw [thin, dashed, gray, fill=red!10] (0,5.3) -- (7,5.3) --
      (10,9.3) -- (3,9.3) -- cycle;
      \draw [thin, dashed, gray, fill=blue!10] (0,0) -- (7,0) --
      (10,4) -- (3,4) -- cycle;
      \draw [thin, dashed, gray] (0,0)  -- (0,6);
      \draw [thin, dashed, gray] (3,4)  -- (3,10); 
      \draw [thin, dashed, gray] (7,0)  -- (7,6);     
      \draw [thin, dashed, gray] (10,4) -- (10,10);
      \tikzstyle{universalnode}= [draw, thick, circle, inner sep=0pt,
        minimum size=5.5em, fill=red!20]
      \node [universalnode] at (0,5.3)  (Aup) {$U_n(\D)$};
      \node [universalnode] at (3,9.3)  (Bup) {$U_n(\Drminustwo)$}; 
      \node [universalnode] at (5,7.3)  (Cup) {$U_n(\Di)$};
      \node [universalnode] at (7,5.3)  (Dup) {$U_n(\Drtwo)$};
      \node [universalnode] at (10,9.3) (Eup) {$U_n(\Domega)$};
      \tikzstyle{permutationnode}= [draw, thick, regular
          polygon, regular polygon sides=8, inner sep=-4pt, minimum
          size=5.5em, fill=blue!20]
      \node [permutationnode] at (0,0)  (A) {$U_n(\Z)$};
      \node [permutationnode] at (5,2)  (C) {$U_n(\Zi)$};
      \node [permutationnode] at (10,4) (E) {$U_n(\Zomega)$};
      \tikzstyle{classicalnode}= [draw, thick, regular polygon,
        regular polygon sides=4, inner sep=0pt, minimum size=5.5em,
        fill=yellow!20]
      \node [classicalnode] at (-3.5,2.45)    (Cl) {$S_n$};
      \draw [->, thick, black] (Aup) -- node[midway,left]  {$F$} (Bup);
      \draw [->, thick, black] (Aup) -- node[midway,above] {$S$}  (Cup);
      \draw [->, thick, black] (Aup) -- node[midway,above] {$CH$} (Dup);
      \draw [->, thick, black] (Bup) -- node[midway,above] {$CH$}  (Eup);
      \draw [->, thick, black] (Cup) -- node[midway,above]  {$T$}  (Eup);
      \draw [->, thick, black] (Dup) -- node[midway,left]  {$F$}  (Eup);
      \draw [->, thick, black] (A)   -- node[midway,below] {$S$}  (C);
      \draw [->, thick, black] (C)   -- node[midway,below] {$T$}  (E);
      \draw [->, thick, black] (A)   -- node[midway,left]  {$H\otimes H$}  (Aup);
      \draw [->, thick, black] (C)   -- node[midway,left]  {$H\otimes H$}  (Cup);
      \draw [->, thick, black] (E)   -- node[midway,right]  {$H\otimes H$}  (Eup);
      \draw [->, thick, black] (Cl)  -- node[midway,left]  {$H\otimes H$}  (Aup);
      \draw [->, thick, black] (Cl)  -- node[midway,below] {$Z$}  (A);
    \end{tikzpicture}
  \end{center}
  \caption{Some subgroups of $U_n(\Domega)$. To the left of the cube,
    in yellow, the symmetric group $S_n$ corresponds to circuits over
    the gate set $\s{X, CX, CCX}$. On the bottom face of the cube, in
    blue, are generalized symmetric groups, and on the top face of the
    cube, in red, are universal subgroups of $U_n(\Domega)$. The edges
    of the lattice denote inclusion. The gates labeling the edges are
    sufficient to extend the expressive power of a gate set from one
    subgroup to the next (and no further). For example, the edge
    labeled $Z$ going from $S_n$ to $U_n(\Z)$ indicates that adding
    the $Z$ gate to $\s{X, CX, CCX}$ produces a gate set expressive
    enough to represent every matrix in $U_n(\Z)$ (but not every
    matrix in $U_n(\Zi)$).}\label{fig:inclusionlattice}
\end{figure}
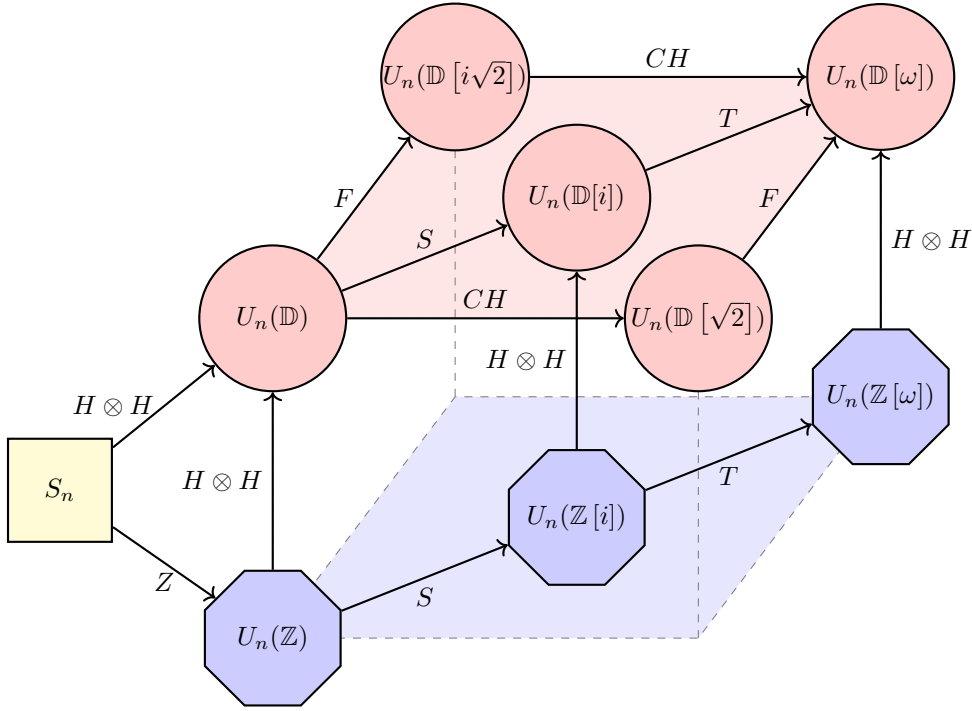

The rest of the paper is organized as follows. In \cref{sec:overview},
we give an overview of our methods. In \cref{sec:ringsandmatrices}, we
introduce the rings and matrices which will be used throughout the
paper. In \cref{sec:circuits}, we show that certain useful matrices
can be exactly represented by restricted Clifford+$T$
circuits. \cref{sec:characterizations} contains the proofs of our
various number-theoretic characterizations. We conclude in
\cref{sec:conc}.

\section{Overview}
\label{sec:overview}

Unrestricted Clifford+$T$ circuits are generated by
\[
  H= \frac{1}{\sqrt{2}}
  \begin{bmatrix}
      1 & 1 \\
      1 & -1
  \end{bmatrix},
\quad
  CX = 
  \begin{bmatrix}
    1 & 0 & 0 & 0 \\
    0 & 1 & 0 & 0 \\
    0 & 0 & 0 & 1 \\
    0 & 0 & 1 & 0 \\
  \end{bmatrix},
\quad
\mbox{ and }
\quad
  T =  
  \begin{bmatrix}
      1 & 0 \\
      0 & \omega
  \end{bmatrix}. 
\]
Since $\omega = (1+i)/\sqrt{2}$, the entries of all the generators
belong to the ring
$\Z[1/\sqrt{2},\omega]=\Z[1/\sqrt{2},i]=\D[\omega]$. Hence, if a
matrix $V$ can be represented exactly by an $n$-qubit Clifford+$T$
circuit, then $V\in U_{2^n}(\Domega)$, the group of $2^n\times 2^n$
unitary matrices with entries in $\Domega$. Showing that the ring
$\Domega$ characterizes Clifford+$T$ circuits thus amounts to proving
the converse implication. An algorithm establishing that every element
of $U_{2^n}(\Domega)$ can be exactly represented by a Clifford+$T$
circuit is known as an exact synthesis algorithm.

The original insight of Kliuchnikov, Maslov and Mosca in the
single-qubit Clifford+$T$ case was to reduce the problem of exact
synthesis to the problem of state preparation. The latter problem is
to find, given a target vector $v\in\Domega^n$, a sequence
$G_1,\ldots, G_\ell$ of Clifford+$T$ gates such that $G_\ell\cdots
G_1e_1 = v$ or, equivalently, such that $G_1^\dagger\cdots
G_\ell^\dagger v=e_1$. Kliuchnikov et al. realized that this sequence
of gates can be found by first writing $v$ as $v=u/\sqrt{2}^q$ for
some $u\in\Zomega$ and then iteratively reducing the exponent $q$.

This basic premise was extended by Giles and Selinger to the
multi-qubit context by adding an outer induction over the columns of
an $n$-qubit unitary. This method amounts to performing a constrained
Gaussian elimination where the row operations are restricted to a few
basic moves. The Giles-Selinger algorithm proceeds by reducing the
leftmost column of an $n\times n$ unitary matrix to the first standard
basis vector by applying a sequence of one- and two-level matrices,
which act non-trivially on at most two components of a vector, before
recursively dealing with the remaining submatrix. If the target
unitary is $V = \mleft[\begin{array}{c|c} v & V' \end{array}\mright]$,
then the Giles-Selinger algorithm first constructs a sequence of
matrices $G_1, \dots, G_\ell$ such that $G_1\cdots G_\ell v =
e_1$. Left-multiplying $V$ by this sequence of matrices then yields
\[
G_1\cdots G_\ell
\mleft[\begin{array}{c|ccc}
  & & & \\ v & & V' & \\ & & &
\end{array}\mright] 
=
\mleft[\begin{array}{c|ccc}
  1 & 0 & \cdots & 0 \\
  \hline
  0 & & & \\
  \vdots & & V'' & \\
  0 & & &
\end{array}\mright] 
\]
where $V''$ is unitary. The fact that the matrices used in this
reduction act non-trivially on no more than two rows of the matrix
ensures that when the algorithm recursively reduces the columns of
$V''$ it does so without perturbing the previously fixed columns. The
Giles-Selinger algorithm thus relies on the following two facts.
\begin{enumerate}
  \item A unit vector in $\Domega^n$ can be reduced to a standard
    basis vector by using one- and two-level matrices and
  \item the required one- and two-level matrices can be exactly
    represented by Clifford+$T$ circuits.
\end{enumerate}
While there are subtle differences between the various cases discussed
below, our method in characterizing restricted Clifford+$T$ circuits
follows this general structure.

\section{Rings and Matrices}
\label{sec:ringsandmatrices}

In this section, we discuss the rings and matrices that will be used
throughout the paper. For further details, the reader is encouraged to
consult \cite{artin}.

\subsection{Rings}
\label{ssec:rings}

We write $\N$ for the set of nonnegative integers and if $n\in\N$ we
write $[n]$ for the set $\s{1,\ldots,n}$. We use $\Z$ to denote the
ring of integers and $i$ to denote the imaginary unit. 
We define $\omega$ as
$\omega=e^{i\pi/4} = (1+i)/\sqrt{2}$. Note that $i$ is a 4-th root of
unity and that $\omega$ is an 8-th root of unity.

We will use the extensions of $\Z$ defined below.

\begin{definition}
  \label{def:rings}
  Let
  \begin{itemize}
  \item $\Zrtwo = \s{x_0+x_1\sqrt{2} \mid x_0,x_1\in \Z}$,
  \item $\Zrminustwo = \s{x_0+x_1i\sqrt{2} \mid x_0,x_1\in \Z}$,
  \item $\Zi = \s{x_0+x_1i \mid x_0,x_1\in \Z}$, and
  \item $\Zomega = \s{x_0+x_1\omega+x_2\omega^2 + x_3\omega^3 \mid
    x_0,x_1,x_2,x_3\in \Z}$.
  \end{itemize}
\end{definition}

The rings $\Zrtwo$, $\Zrminustwo$, $\Zi$, and $\Zomega$ are known as
the ring of \emph{quadratic integers with radicand 2}, the ring of
\emph{quadratic integers with radicand \mbox{-}2}, the ring of
\emph{Gaussian integers}, and the ring of \emph{cyclotomic integers of
  degree 8}, respectively. All of these rings are distinct subrings of
$\Zomega$ and we have the inclusions depicted in the lattice of
subrings below.
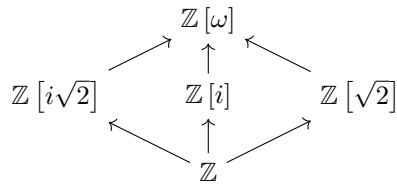
\begin{center}
\begin{tikzpicture}
  \node at (1,0) (r1) {$\Z$};
  \node at (-1,1) (r2) {$\Zrminustwo$};
  \node at (1,1) (r3) {$\Zi$};
  \node at (3,1) (r4) {$\Zrtwo$};
  \node at (1,2) (r5) {$\Zomega$};
  \draw [->] (r1) -- (r2);
  \draw [->] (r1) -- (r3);
  \draw [->] (r1) -- (r4);
  \draw [->] (r2) -- (r5);
  \draw [->] (r3) -- (r5);
  \draw [->] (r4) -- (r5);    
\end{tikzpicture}
\end{center}

Further to the rings introduced in \cref{def:rings}, we will consider
extensions of the ring of \emph{dyadic fractions}, i.e., fractions
whose denominator is a power of 2.

\begin{definition}
  \label{def:dyadicfractions}
  The ring of \emph{dyadic fractions} $\D$ is defined as $\D=\left\{
  \frac{u}{2^q} \mid u\in \Z, q\in\N \right\}$.
\end{definition}

\begin{definition}
  \label{def:dyadicrings}
  Let
  \begin{itemize}
  \item $\Drtwo = \s{x_0+x_1\sqrt{2} \mid x_0,x_1\in \D}$,
  \item $\Drminustwo = \s{x_0+x_1i\sqrt{2} \mid x_0,x_1\in \D}$,
  \item $\Di = \s{x_0+x_1i \mid x_0,x_1\in \D}$, and
  \item $\Domega = \s{x_0+x_1\omega+x_2\omega^2 + x_3\omega^3 \mid
    x_0,x_1,x_2,x_3\in \D}$.
  \end{itemize}
\end{definition}

If $v\in\Drtwo$, then $v$ can be written as $v= u/2^q$ for some $q\in
\N$ and some $u \in \Zrtwo$. A similar property holds for elements of
$\Drminustwo$, $\Di$, and $\Domega$.

If $R$ is a ring and $r\in R$ we write $R/(r)$ for the quotient of the
ring $R$ by the ideal generated by the element $r$. Two elements $s$
and $s'$ of $R$ are congruent modulo $r$ if $s-s'$ belongs to the
ideal $(r)$, in which case we write $s\equiv s' \pmod{r}$. We
sometimes refer to the elements of the ring $R/(r)$ as residues. Some
quotient rings are well-known. For example, $\Z/(2) = \s{0,1}$ and
$\Z/(4)=\s{0,1,2,3}$. The following proposition gives an explicit
description of certain lesser-known rings of residues which will be
useful in what follows.

\begin{proposition}
  \label{prop:ringsofresidues}
  We have
  \begin{itemize}
  \item $\Zrtwo/(2) = \s{0,1,\sqrt{2},1+\sqrt{2}}$, 
  \item $\Zrminustwo/(2) =
    \s{0,1,i\sqrt{2},1+i\sqrt{2}}$,
  \item $\Zrminustwo/(2i\sqrt{2}) =
    \s{0,1,2,3,i\sqrt{2},1+i\sqrt{2},
      2+i\sqrt{2}, 3+i\sqrt{2}}$, and
  \item $\Zi/(2) = \s{0,1,i,1+i}$.    
  \end{itemize}
\end{proposition}

\begin{proof}
  To see, for example, that $\Zrtwo/(2) =
  \s{0,1,\sqrt{2},1+\sqrt{2}}$, note that $u=x_0+x_1\sqrt{2}$ and
  $u'=x_0'+x_1'\sqrt{2}$ are congruent modulo 2 if there exists an
  element $t= y_0+y_1\sqrt{2}$ such that $u-u'=2t$. This is the case
  if and only if $(x_0-x_0')+(x_1-x_1')\sqrt{2} = 2y_0+2y_1\sqrt{2}$
  which in turn holds if and only if $x_0\equiv x_0' \pmod{2}$ and
  $x_1\equiv x_1' \pmod{2}$.
\end{proof}

We will often take advantage of properties of residues. Some of these
properties are generic. For example, if $u$ and $v$ are two elements
of a ring $R$ and $u\equiv v\pmod{2}$, then $u\pm v \equiv 0
\pmod{2}$. Other properties of residues are specific to a particular
ring. For example, an integer $u\in\Z$ is odd if and only if
$u^2\equiv 1 \pmod{4}$. Similarly, for an integer $u\in\Z$, we have
$u\equiv 3 \pmod{4}$ if and only if $-u\equiv 1\pmod{4}$. We now state
important properties of residues in $\Zrminustwo$ and $\Zi$ for future
reference. They can be established by reasoning using residue tables
in the relevant quotient rings. In the following, we denote the
complex conjugate of an element $u$ by $u^\dagger$. For uniformity, we
sometimes write $u^\dagger$ even when $u$ belongs to a real subring of
$\Domega$. In this case, $u^\dagger = u$.

\begin{proposition} 
  \label{prop:zrminus2residues}
  The following statements hold.
    \begin{itemize}
    \item In $\Zrminustwo/(2)$, $u^\dagger u \equiv 0 \mbox{ or } 1$.
    \item If $u^\dagger u \equiv 1$ in $\Zrminustwo/(2)$, then
      $u\equiv 1,3,1+i\sqrt{2},\mbox{ or }3+i\sqrt{2}$ in
      $\Zrminustwo/(2i\sqrt{2})$.
  \item In $\Zrminustwo/(2i\sqrt{2})$, $u\equiv 3$ if and only if
    $-u\equiv 1$ and $u\equiv 3+i\sqrt{2}$ if and only if $-u\equiv
    1+i\sqrt{2}$.
  \end{itemize}
\end{proposition}

\begin{proposition}
  \label{prop:ziresiduesmod2}
  The following statements hold.
  \begin{itemize} 
  \item In $\Zi/(2)$, if $u^2 \equiv 1$, then $u\equiv 1 \mbox{ or } i$.
  \item In $\Zi/(2)$, $u\equiv i$ if and only if $iu\equiv 1$.
  \end{itemize}
\end{proposition}

\subsection{Matrices}
\label{ssec:matrices}

We write $e_j$ for the $j$-th standard basis vector and $M^\dagger$
for the conjugate transpose of the matrix $M$. If $R$ is a ring, we
sometimes write $R^{n\times n'}$ for the collection of $n\times n'$
matrices over $R$. We will use one-, two-, and four-level matrices
which act non-trivially on only one, two, or four of the components of
their input. These matrices will be defined using basic matrices. The
construction is best explained with an example. If
\[
V = \begin{bmatrix} v_{1,1} & v_{1,2} \\ v_{2,1} &
  v_{2,2} \end{bmatrix}
\]
is a 2-dimensional unitary matrix, then in 3 dimensions the two-level
operator of type $V$, denoted by $V_{[1,3]}$, is the matrix given
below.
\[
V_{[1,3]} = \begin{bmatrix} v_{1,1} & 0 & v_{1,2} \\ 0 & 1 & 0
  \\ v_{2,1} & 0 & v_{2,2} \end{bmatrix}
\]

\begin{definition}
  \label{def:onetwolevel}
  Let $W$ be an $n\times n$ unitary matrix, let $n\leq n'$, and let
  $a_1,\ldots,a_n\in[n']$. The \emph{$n$-level matrix of type $W$} is
  the $n'\times n'$ unitary matrix $W_{[a_1,\ldots,a_n]}$ defined by
  \[
  {W_{[a_1,\ldots,a_n]}}_{j,k} = \begin{cases} W_{j',k'} \mbox{ if }
    j=a_{j'} \mbox{ and } k = a_{k'}\\ I_{j,k} \mbox{
      otherwise.} \end{cases}
  \]
\end{definition}

Let $R$ be one of $\Z$, $\Zrtwo$, $\Zrminustwo$, $\Zi$ or $\Zomega$
and let $p$ be an element of $\Zomega$. We will be interested in
matrices of the form
\begin{equation}
  \label{eq:generic}
  V=\frac{1}{p^q}W
\end{equation}
where $W$ is a matrix over $R$ and $q\in\N$.

\begin{definition}
  \label{def:denomexp}
  Fix $R\in \s{\Z, \Zrtwo, \Zrminustwo, \Zi, \Zomega}$. If $V$ is a
  matrix of the form \eqref{eq:generic} and $q'\in\N$, then we say
  that $q'$ is a \emph{$p$-denominator exponent} of $V$ if
  \[
  p^{q'}V\in R^{m\times n}.
  \]
  The smallest such $q'$ is the \emph{least $p$-denominator exponent} of
  $V$, denoted $\lde_p(V)$.
\end{definition}

We sometimes omit $p$ when the base of the exponent is clear from the
context. Note that the notion of denominator exponent applies to
matrices of any dimension and we can therefore talk about the
denominator exponent of a vector or scalar.

\section{Circuits}
\label{sec:circuits}

In this section, we review basic circuit constructions which will be
useful below. A more detailed discussion of quantum circuits can be
found in Chapter 4 of \cite{NC}.

Let $\zeta$ be an $m$-th root of unity. We sometimes call $\zeta$ a
\emph{global phase of order $m$}. We think of these global phases as
gates acting on 0 qubits and in what follows we will be especially
interested in the global phases of order 2, 4, and 8, which we denote
$-1$, $i$, and $\omega$, respectively.  The single-qubit \emph{phase
  gate of order $m$} is defined as
\[
  P_\zeta =
  \begin{bmatrix}
      1 & 0 \\
      0 & \zeta
    \end{bmatrix}.
\]
We will be particularly interested in phase gates of order 2, 4, and 8
which we call the $Z$, $S$, and $T$ gates, respectively. Hence
\[
  Z=
  \begin{bmatrix}
      1 & 0 \\
      0 & -1
  \end{bmatrix},
\quad
  S = 
  \begin{bmatrix}
      1 & 0 \\
      0 & i
  \end{bmatrix},
\quad
\mbox{ and }
\quad
  T =  
  \begin{bmatrix}
      1 & 0 \\
      0 & \omega
  \end{bmatrix}. 
\]
In addition to phase gates, we will also use the single-qubit gates
$H$ and $X$ defined by
\[
  H  = 
  \frac{1}{\sqrt 2}\begin{bmatrix}
    1 & 1 \\
    1 & -1
  \end{bmatrix}
  \quad
  \mbox{ and }
  \quad
  X = 
  \begin{bmatrix}
    0 & 1 \\
    1 & 0
  \end{bmatrix}.
\]
The $H$ gate is the \emph{Hadamard} gate and the $X$ gate is the
\emph{NOT} gate.  The last single-qubit gate we will use is the $F$
gate defined below.
\[
  F  = 
  \frac{1}{2}\begin{bmatrix}
    1 + i\sqrt{2} & 1 \\
    1 & -1 + i\sqrt{2}
  \end{bmatrix}.
\]
The $F$ gate is not as common as the other single-qubit gates
introduced above. We note that $F^2 = iH$ and that $F$ can be
expressed as a product of better-known gates in Matsumoto-Amano normal
form \cite{MA08},
\[
  F = SHTSHTSHS\omega^{-1}.
\]
We will also make use of the two-qubit $H\otimes H$ gate as well as
the \emph{controlled} gates defined below.
\[
  CH =
  I_2 \oplus H,
\qquad
  CX =
  I_2 \oplus X,
  \quad
  \mbox{ and }
\qquad
  CCX = I_6 \oplus X.
\]
We will refer to these gates as the \emph{controlled}-$H$ gate,
\emph{controlled}-$X$ or \emph{CNOT} gate, and the
\emph{doubly-controlled}-$X$ or \emph{Toffoli} gate, respectively. In
general, if $G$ is a gate, then we write $C^nG$ for the
\emph{$n$-control}-$G$ gate.

As usual, \emph{circuits} are built from gates through composition and
tensor product. An \emph{ancilla} is a qubit used locally within a
circuit but on which the global action of the circuit is trivial.  We
say that a $2^n\times 2^n$ unitary matrix $W$ is exactly represented
by a circuit $D$ using $m$ \emph{clean} ancillas if for any $n$-qubit
state $\ket{\psi}$,
\[
  D\ket{\psi}\ket{0}^{\otimes m} = (W\ket{\psi})\ket{0}^{\otimes m}.
\]
If the circuit is independent of the initial state of the ancilla, it
is said to accept a \emph{dirty} ancilla. In particular, we say that a
$2^n\times 2^n$ unitary matrix $W$ is exactly represented by a circuit
$D$ using $m$ \emph{dirty} ancillas if for any $n$-qubit state
$\ket{\psi}$ and any $m$-qubit state $\ket{\phi}$,
\[
  D\ket{\psi}\ket{\phi} = (W\ket{\psi})\ket{\phi}.
\]
Note that a clean ancilla can always be used in place of a dirty
ancilla.

In order to characterize restricted Clifford+$T$ circuits, it is
helpful to establish some basic facts about the construction of
multi-level matrices over gate sets including the Toffoli gate.

\begin{proposition}
  \label{lem:perm}
  Any $2^n\times 2^n$ permutation matrix $V$ can be exactly
  represented by a circuit over the gate set $\{X, CX, CCX \}$ with at
  most one dirty ancilla.
\end{proposition}

\begin{proof}
  By the Giles-Selinger algorithm \cite{GS13} restricted to
  permutation matrices, $V$ can be represented by a circuit over
  two-level $X$ gates. Likewise, each two-level $X$ gate can be
  implemented over fully-controlled $X$ and single-qubit gates by
  using a Gray code (see, e.g., \cite[Sec.~4.5.2]{NC}). Each
  fully-controlled $X$ gate can be implemented with one dirty ancilla
  by \cite{BBC95}, which completes the proof.
\end{proof}

\begin{proposition}
  \label{lem:foo}
  Let $W$ be a $2^m\times 2^m$ unitary matrix and let $\mathcal{G}$ be
  a set of gates. If $CW$ can be exactly represented over $\{X, CX,
  CCX\} \cup \mathcal{G}$ using at most one dirty ancilla, then, for
  any $n\geq 1$, $C^nW$ can also be exactly represented over $\{X, CX,
  CCX\}\cup\mathcal{G}$. Moreover, a single clean ancilla suffices.
\end{proposition}

\begin{proof}
  This follows from standard techniques, e.g. \cite{BBC95}.  If $n =
  1$, then $CW$ can be implemented with a single dirty ancilla and
  thus also with a clean one. If $n> 1$, then the $C^n W$ gate can be
  implemented as follows, where each gate on the right has at least
  one dirty ancilla available for use.
  
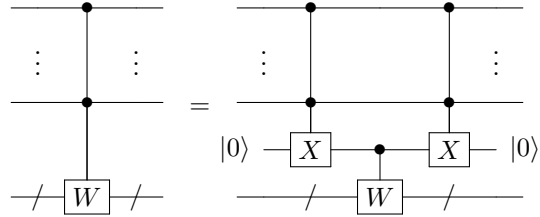
\[
  \Qcircuit @C=1em @R=.5em @!R {
          & \qw & \ctrl{4} & \qw & \qw \\
          & \vdots  & & \vdots & \\
          & \qw & \ctrl{2} & \qw & \qw  \\
          &  &  &  & \\
          & {/} \qw & \gate{W} & {/} \qw & \qw
  }
  \raisebox{-3.7em}{\quad=\quad}
  \Qcircuit @C=1em @R=.5em @!R {
          & \qw & \ctrl{3} & \qw & \ctrl{3} & \qw & \qw \\
          & \vdots  & & & & \vdots & \\
          & \qw & \ctrl{1} & \qw & \ctrl{1} & \qw & \qw  \\
          & \lstick{\ket{0}} & \gate{X}  & \ctrl{1} & \gate{X} & \qw & \rstick{\hspace{-1em}\ket{0}} \\
          & \qw & {/} \qw & \gate{W} & {/} \qw & \qw & \qw
  }
  \]
\end{proof}

\begin{corollary}
  \label{lem:multilevel}
  Let $W$ be a $2^m\times 2^m$ unitary matrix and let $\mathcal{G}$ be
  a set of gates.  If $CW$ can be exactly represented over $\{X, CX,
  CCX\} \cup \mathcal{G}$ with at most one dirty ancilla, then
  $W_{[a_1,a_2\dots, a_{2^m}]}$ is representable over $\{X, CX, CCX\}
  \cup \mathcal{G}$ with at most one clean ancilla.
\end{corollary}

\begin{proof}
  Follows from \cref{lem:perm,lem:foo} by noting that there exists a
  $2^{n+m}$-dimensional permutation matrix $V$ such that
  \[
    W_{[a_1,a_2\dots, a_{2^m}]} = V^\dagger(C^nW)V.
  \]
\end{proof}

We can now use \cref{lem:multilevel} to give constructions of
multi-level matrices of different types over their uncontrolled
versions in the presence of the Toffoli gate.

\begin{proposition}
  \label{prop:intcircs}
  The operators
  \[
  \s{(-1)_{[a]},X_{[a,b]},(H\otimes H)_{[a,b,c,d]}},
  \]
  where $a$, $b$, $c$, and $d$ are distinct elements of $[n]$, can be
  exactly represented by quantum circuits over the gate set $\s{X, CX,
    CCX, H\otimes H}$ using at most one clean ancilla.
\end{proposition}

\begin{proof}
  By \cref{lem:multilevel} it suffices to give constructions for the
  singly-controlled $Z$ and $H\otimes H$ gates using at most 
  a single dirty ancilla. We have 
  
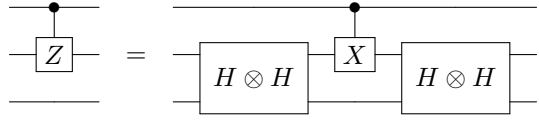
\[
  \Qcircuit @C=1em @R=.5em @!R {
          & \ctrl{1} & \qw \\
          & \gate{Z} & \qw \\
          & \qw & \qw
  }
  \raisebox{-1.8em}{\quad=\quad}
  \Qcircuit @C=1em @R=.5em @!R {
          & \qw & \ctrl{1} & \qw & \qw \\
          & \multigate{1}{H\otimes H} & \gate{X} & \multigate{1}{H\otimes H} & \qw \\
          & \ghost{H\otimes H} & \qw & \ghost{H\otimes H} & \qw
  }
  \]
  and it can be verified that the equality below holds.
  
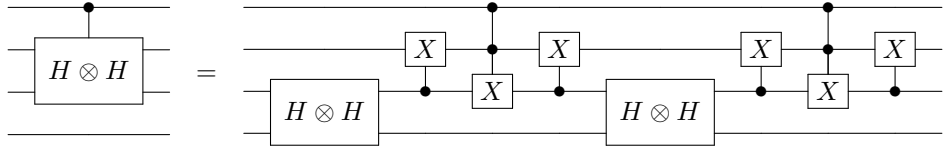
\[
  \Qcircuit @C=1em @R=.7em @!R {
          & \ctrl{1} & \qw \\
          & \multigate{1}{H\otimes H} & \qw \\
          & \ghost{H\otimes H} & \qw \\
          & \qw & \qw
  }
  \raisebox{-2.5em}{\quad=\quad}
  \Qcircuit @C=1em @R=.3em @!R {
          & \qw & \qw & \ctrl{2} & \qw & \qw & \qw & \ctrl{2} & \qw & \qw \\
          & \qw & \gate{X} & \ctrl{1} & \gate{X} & \qw & \gate{X} & \ctrl{1} & \gate{X} & \qw \\
          & \multigate{1}{H\otimes H} & \ctrl{-1} & \gate{X} & \ctrl{-1} & \multigate{1}{H\otimes H} & \ctrl{-1} & \gate{X} & \ctrl{-1} & \qw \\
          & \ghost{H\otimes H} & \qw & \qw & \qw & \ghost{H\otimes H} & \qw & \qw & \qw & \qw
  }
  \]
\end{proof}

\begin{corollary}
  \label{cor:intcircs}
  The operators
  \[
  \s{(-1)_{[a]},X_{[a,b]},(H\otimes H)_{[a,b,c,d]}, I_{2^{n-1}}\otimes
    H},
  \]
  where $a$, $b$, $c$, and $d$ are distinct elements of $[n]$, can be
  exactly represented by quantum circuits over the gate set $\s{X, CX,
    CCX, H}$ using at most one clean ancilla.
\end{corollary}

\begin{proposition}
  \label{prop:gaussiancircs}
  The operators
  \[
  \s{i_{[a]},X_{[a,b]},\omega H_{[a,b]}},
  \]
  where $a$ and $b$ are distinct elements of $[n]$, can be exactly
  represented by quantum circuits over the gate set $\s{X, CX, CCX,
    \omega H, S}$ using at most one clean ancilla.
\end{proposition}

\begin{proof}
  Again, it suffices to give constructions for the singly-controlled
  $S$ and $\omega H$ gates.  In this case it can be verified that both
  of the equalities below hold.
  
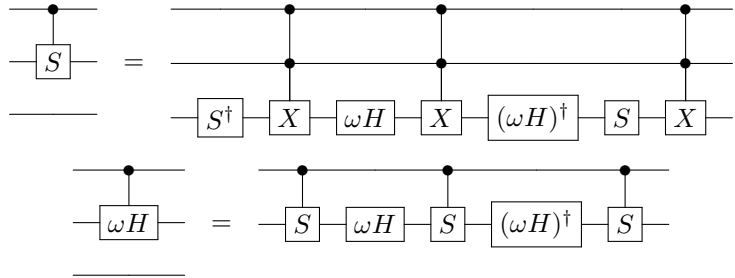
\[
  \Qcircuit @C=1em @R=.7em @!R {
          & \ctrl{1} & \qw \\
          & \gate{S} & \qw \\
          & \qw & \qw
  }
  \raisebox{-2em}{\quad=\quad}
  \Qcircuit @C=1em @R=.35em @!R {
          & \qw & \ctrl{2} & \qw & \ctrl{2} & \qw & \qw & \ctrl{2} & \qw \\
          & \qw & \ctrl{1} & \qw & \ctrl{1} &\qw & \qw & \ctrl{1} & \qw \\
          & \gate{S^\dagger} & \gate{X} & \gate{\omega H} & \gate{X} & \gate{(\omega H)^\dagger} & \gate{S} & \gate{X} & \qw
  }
  \]
  \[
  \Qcircuit @C=1em @R=.7em @!R {
          & \ctrl{1} & \qw \\
          & \gate{\omega H} & \qw \\
          & \qw & \qw
  }
  \raisebox{-2em}{\quad=\quad}
  \Qcircuit @C=1em @R=.35em @!R {
          & \ctrl{1} & \qw & \ctrl{1} & \qw & \ctrl{1} & \qw \\
          & \gate{S} & \gate{\omega H} & \gate{S} & \gate{(\omega H)^\dagger} & \gate{S} & \qw \\
          & \qw & \qw & \qw & \qw & \qw & \qw 
  }
  \]
\end{proof}

\begin{corollary}
  \label{cor:gaussiancircs}
  The operators
  \[
  \s{i_{[a]},X_{[a,b]},\omega H_{[a,b]}, \omega I_n},
  \]
  where $a$ and $b$ are distinct elements of $[n]$, can be exactly
  represented by quantum circuits over the gate set $\s{X, CX, CCX, H,
    S}$ using at most one clean ancilla.
\end{corollary}

\begin{proof}
  Follows from \cref{prop:gaussiancircs} and the fact that $\omega =
  SHSHSH$.
\end{proof}

\begin{proposition}
  \label{prop:realcircs}
  The operators
  \[
  \s{(-1)_{[a]},X_{[a,b]},H_{[a,b]}},
  \]
  where $a$ and $b$ are distinct elements of $[n]$, can be exactly
  represented by quantum circuits over the gate set $\s{X, CX, CCX, H,
    CH}$ using at most one clean ancilla.
\end{proposition}

\begin{proof}
  By \cref{prop:intcircs}, $(-1)_{[a]}$ can be represented by a
  quantum circuit over $\s{X, CX, CCX, H\otimes H}$ and hence also
  $\s{X, CX, CCX, H, CH}$. Since $CH$ is already in the generating set
  the proof is complete.
\end{proof}

\begin{proposition}
  \label{prop:unrealcircs}
  The operators
  \[
  \s{(-1)_{[a]},X_{[a,b]},F_{[a,b]}},
  \]
  where $a$ and $b$ are distinct elements of $[n]$, can be exactly
  represented by quantum circuits over the gate set $\s{X, CX, CCX,
    F}$ using at most one clean ancilla.
\end{proposition}

\begin{proof}
  To show that $CZ$ is representable over the gate set, it can be
  observed that the equality below holds, since $F^2 = iH$ and $F^6=-iH$.
  \[
  \Qcircuit @C=1em @R=.7em @!R {
          & \ctrl{1} & \qw \\
          & \gate{Z} & \qw
  }
  \raisebox{-1em}{\quad=\quad}
  \Qcircuit @C=1em @R=.7em @!R {
          & \qw & \ctrl{1} & \qw & \qw \\
          & \gate{F^2} & \gate{X} & \gate{F^6} & \qw
  }
  \]
  The construction of $CF$ is somewhat more involved, but can be
  obtained from standard constructions (e.g., \cite{BBC95}) by first noting
  that $(ZXF)^2$ and $X(ZXF)X(ZXF)X = ZXF$. The $CF$ gate can then be constructed as below.
  \[
  \Qcircuit @C=1em @R=.7em @!R {
          & \ctrl{1} & \qw \\
          & \gate{F} & \qw
  }
  \raisebox{-1em}{\quad=\quad}
  \Qcircuit @C=1em @R=.7em @!R {
          & \ctrl{1} & \ctrl{1} & \qw & \qw & \qw & \qw & \ctrl{1} & \qw & \qw & \qw & \qw & \qw \\
          & \gate{X} & \gate{Z} & \gate{X} & \gate{Z} & \gate{X} & \gate{F} & \gate{X} & \gate{Z} & \gate{X} & \gate{F} & \gate{X} & \qw
  }
  \]
\end{proof}

\section{Number-Theoretic Characterizations}
\label{sec:characterizations}

\subsection{The \texorpdfstring{$\D$}{D} case}
\label{ssec:integral}

We start by studying the group of $n\times n$ unitary matrices over
$\D$. Since $X$, $CX$, $CCX$, and $H\otimes H$ have entries in $\D$,
any circuit over the gate set $\s{X, CX, CCX, H\otimes H}$ must
represent a unitary matrix over $\D$. Here, we show the converse: any
unitary matrix over $\D$ can be represented by a circuit over $\s{X,
  CX, CCX, H\otimes H}$. To prove this, it is sufficient to establish
that every unitary over $\D$ can be expressed as a product of the
following generators
\begin{equation}
  \label{eq:evenintgens}
  \s{(-1)_{[a]},X_{[a,b]},(H\otimes H)_{[a,b,c,d]}},
\end{equation}
where $a$, $b$, $c$, and $d$ are distinct elements of $[n]$. Indeed,
by \cref{prop:intcircs}, all of the above generators can be exactly
represented by quantum circuits over the gate set $\s{X, CX, CCX,
  H\otimes H}$.

If $V$ is a matrix over $\D$, then $V$ can be written as
\begin{equation}
  \label{eq:integralmatrix}
  V=\frac{1}{2^q}W
\end{equation}
where $q\in\N$ and $W$ is a matrix over $\Z$. We will consider 
$2$-denominator exponents of such matrices.

The following four lemmas are devoted to proving the analogue of Giles
and Selinger's \emph{Column Lemma} (Lemma 5 in \cite{GS13}). Here, the
goal is to establish that any unit vector over $\D$ can be reduced to
a standard basis vector by multiplying it on the left by an
appropriately chosen sequence of generators. We consider the case of
vectors of dimension $n<4$ first, before moving on to higher
dimensions.

\begin{lemma}
  \label{lem:integralcolumnevendenomlessthanfour}
  Let $n <4$ and let $j\in [n]$. If $v$ is an $n$-dimensional unit
  vector over $\D$, then there exist generators $G_1, \ldots, G_\ell$
  from \eqref{eq:evenintgens} such that $G_1\cdots G_\ell v = e_j$.
\end{lemma}

\begin{proof}
  Write $v$ as $v=u/2^q$ with $u\in\Z^n$ and $q=\lde_2(v)$. Since $v$
  is a unit vector, we have $v^\dagger v = 1$ and thus $4^q=\sum
  u_k^\dagger u_k =\sum u_k^2$. The square of any odd number is
  congruent to 1 modulo 4. Thus when $n<4$, we have $\sum u_k^2\equiv
  0 \pmod{4}$ only if every $u_k$ is even. This implies that
  $\lde_2(v)=0$ when $n<4$ and therefore that $v=\pm e_{j'}$ for some
  $j'\in[n]$. Hence one of
    \[
    v=e_j, \quad (-1)_{[j]}v=e_j, \quad X_{[j,j']}v = e_j, \quad
      \mbox{or} \quad X_{[j,j']}(-1)_{[j']}v = e_j
    \]
    must hold, which completes the proof.
\end{proof}  

Because $(H\otimes H)_{[a,b,c,d]}$ is a four-level matrix, we consider
its action on certain 4-dimensional vectors in the lemma below. This
is in contrast with Giles and Selinger's algorithm, for which only
one- and two-level matrices are needed.

\begin{lemma}
  \label{lem:integralvector}
  If $u_1,\ldots,u_4\in\Z$ are such that $u_1^2 \equiv \ldots \equiv
  u_4^2 \equiv 1 \pmod{4}$, then there exist $m_1,\ldots,m_4$ such
  that
  \[
  (H \otimes H)(-1)^{m_1}_{[1]}(-1)^{m_2}_{[2]}(-1)^{m_3}_{[3]}
  (-1)^{m_4}_{[4]} \begin{bmatrix} u_1 \\ u_2 \\ u_3
    \\ u_4 \end{bmatrix} = \begin{bmatrix} u_1' \\ u_2' \\ u_3'
    \\ u_4' \end{bmatrix}
  \]
  for some $u_1',\ldots,u_4'\in\Z$ such that $u_1'\equiv\ldots\equiv
  u_4'\equiv 0 \pmod{2}$.
\end{lemma}

\begin{proof}
  If $u\in\Z$ is such that $u^2\equiv 1 \pmod{4}$, then $u\equiv 1
  \pmod{4}$ or $u\equiv 3 \pmod{4}$. Furthermore, if $u\equiv 3
  \pmod{4}$, then $-u\equiv 1\pmod{4}$. Hence, given
  $u_1,\ldots,u_4\in\Z$ such that $u_1^2 \equiv \ldots \equiv u_4^2
  \equiv 1 \pmod{4}$, we can find $m_1,\ldots,m_4$ such that
  $(-1)^{m_1}u_1 \equiv \ldots \equiv (-1)^{m_4}u_4 \equiv 1
  \pmod{4}$. It can then be verified that
  \[
  (H \otimes H) \begin{bmatrix} (-1)^{m_1}u_1 \\ (-1)^{m_2}u_2
    \\ (-1)^{m_3}u_3 \\ (-1)^{m_4}u_4 \end{bmatrix} =
  \begin{bmatrix} u_1' \\ u_2' \\ u_3' \\ u_4' \end{bmatrix}
  \]
  for some $u_1'\equiv\ldots\equiv u_4'\equiv 0 \pmod{2}$.
\end{proof}

\begin{lemma}
  \label{lem:evendenommodfour}
  Let $n \geq 4$. If $v$ is an $n$-dimensional unit vector over $\D$
  and $\lde_2(v)>0$, then there exist generators $G_1, \ldots,
  G_\ell$ from \eqref{eq:evenintgens} such that $G_1\cdots G_\ell v =
  v'$ and $\lde_2(v')<\lde_2(v)$.
\end{lemma}

\begin{proof}
  Write $v$ as $v=u/2^q$ where $u\in\Z^n$ and $q>1$. Since $v$ is a
  unit vector we have $v^\dagger v=1$ and thus $4^q = \sum u_k^\dagger
  u_k = \sum u_k^2$ since $u$ is real. The number of $u_k$ such that
  $u_k^2\equiv 1 \pmod{4}$ is therefore congruent to 0 modulo
  4. Hence, we can group these entries in sets of size 4 and apply
  \cref{lem:integralvector} to each such set in order to reduce the
  2-denominator exponent of the vector.
\end{proof}

\begin{lemma}
  \label{lem:integralcolumnevendenom}
  Let $j\in [n]$. If $v$ is an $n$-dimensional unit vector over $\D$,
  then there exist generators $G_1, \ldots, G_\ell$ from
  \eqref{eq:evenintgens} such that $G_1\cdots G_\ell v = e_j$.
\end{lemma}

\begin{proof}
  The case of vectors of dimension $n<4$ was treated in
  \cref{lem:integralcolumnevendenomlessthanfour} so we assume that
  $n\geq 4$ and we proceed by induction on the least 2-denominator
  exponent of $v$.
  \begin{itemize}
  \item If $\lde_2(v)=0$, then $v$ is a unit vector in $\Z^n$. Hence
    $v=\pm e_{j'}$ for some $j'\in [n]$ and one of
    \[
    v=e_j, \quad (-1)_{[j]}v=e_j, \quad X_{[j,j']}v = e_j, \quad
      \mbox{or} \quad X_{[j,j']}(-1)_{[j']}v = e_j
    \]
    must hold.
  \item If $\lde_2(v)>0$, apply \cref{lem:evendenommodfour} to reduce
    the 2-denominator exponent of $v$. \qedhere
  \end{itemize}
\end{proof}

We can now use \cref{lem:integralcolumnevendenom} to prove that every
unitary matrix with entries in $\D$ can be written as a product of
generators. This, together with \cref{prop:intcircs} establishes our
characterization of circuits over the gate set $\s{X, CX,
  CCX, H \otimes H}$.

\begin{theorem}
  \label{thm:even}
  If $V$ is an $n$-dimensional unitary matrix with entries in $\D$,
  then there exist generators $G_1, \ldots, G_\ell$ from
  \eqref{eq:evenintgens} such that $G_1 \cdots G_\ell V = I$.
\end{theorem}

\begin{proof}
  By iteratively applying \cref{lem:integralcolumnevendenom} to the
  columns of $V$.
\end{proof}

\begin{corollary}
  \label{cor:intevencharac}
  A matrix $V$ can be exactly represented by an $n$-qubit circuit over
  $\s{X, CX, CCX, {H\otimes H}}$ if and only if $V\in
  U_{2^n}(\D)$. Moreover, a single ancilla always suffices to
  construct a circuit for $V$.
\end{corollary}

\subsubsection{Super-integral Clifford+\texorpdfstring{$T$}{T} operators}

One might wonder whether the $H \otimes H$ gate can be replaced with
the $H$ gate.  Since the $H$ gate lies strictly outside of
$U_{2^n}(\D)$, having denominator $\sqrt{2}$, we get a slightly more
general gate set. Circuits over this gate set generate matrices of the
form
\begin{equation}
  \label{eq:integralmatrixalt}
  V=\frac{1}{\sqrt{2}^q}W
\end{equation}
where $q\in\N$ and $W$ is a matrix over $\Z$.

We now leverage \cref{thm:even} and \cref{cor:intcircs} to show that
every unitary matrix of the form in \cref{eq:integralmatrixalt} can be
represented by a circuit over $\s{X, CX, CCX, H }$. For these
matrices, we use $\sqrt{2}$-denominator exponents. We extend the set
of generators from \eqref{eq:evenintgens} with a matrix of the form
$I\otimes H$. Thus the relevant generators are now
\begin{equation}
  \label{eq:intgens}
  \s{(-1)_{[a]},X_{[a,b]},(H\otimes H)_{[a,b,c,d]}, I_{n/2}\otimes H}
\end{equation}
where $a$, $b$, $c$, and $d$ are distinct elements of $[n]$, and
$I_{n/2}\otimes H$ is only well-defined when $n$ is even.  As the
extra generator is only available in even dimensions, we start by
showing that there are no odd-dimension unitary matrices of the form
of \cref{eq:integralmatrixalt} with odd $q$. The proof of this fact is
due to Xiaoning Bian \cite{bian}.

\begin{lemma}
  \label{lem:integraldenomexpprop}
  If $V\neq 0$ is as in \eqref{eq:integralmatrixalt}, then all the
  $\sqrt{2}$-denominator exponents of $V$ are congruent modulo~2.
\end{lemma}

\begin{proof}
  Suppose that $q<q'$ are two $\sqrt{2}$-denominator exponents of
  $V$. Then $V= W/\sqrt{2}^q =W'/\sqrt{2}^{q'}$ for some integer
  matrices $W$ and $W'$. Assume without loss of generality that
  $q<q'$. Then
  \[
  W'=\sqrt{2}^{q'}V=\sqrt{2}^{q'-q}W
  \]
  so that $\sqrt{2}^{q'-q}W$ is an integer matrix. Hence $q\equiv q'
  \pmod{2}$, since $V\neq 0$ and $\sqrt{2}\notin\Z$.
\end{proof}

\begin{lemma}
  \label{lem:columnodd}
  If $v$ is an $n$-dimensional unit vector of the form $v =
  (1/\sqrt{2})^qu$ where $u$ is an integer vector and $q$ is odd, then
  there exist generators $G_1,\ldots, G_\ell$ from \eqref{eq:intgens}
  such that
  \[
  G_1\cdots G_\ell v = \frac{1}{\sqrt{2}}\begin{bmatrix} 1 \\ 1 \\ 0 \\ \vdots \\ 0 \end{bmatrix}.
  \]
\end{lemma}

\begin{proof}
  By induction on $q$. For the base case, it suffices to observe that
  since $\sum u_k^2 = 2$, there exist $j, j'$ such that $u_j =
  (-1)^{m_j}$ and $u_{j'} = (-1)^{m_{j'}}$ and the other entries of
  $u$ are all 0.  It can then be verified that
  \[
    X_{[0,j]}X_{[1,j']}(-1)^{m_j}_{[j]}(-1)^{m_{j'}}_{[j']}v =
    \frac{1}{\sqrt{2}}\begin{bmatrix} 1 \\ 1 \\ 0 \\ \vdots
      \\ 0 \end{bmatrix}.
  \]
  Now suppose $q\geq 3$. Then $\sum u_k^2 = 2^q=4^{q'}$ and hence the
  number of $u_k^2\equiv 1 \pmod{4}$ is congruent to $0$ modulo $4$.
  Then, as in \cref{lem:evendenommodfour}, we can group these entries
  in sets of size $4$ and apply \cref{lem:integralvector} to each set
  to reduce the $\sqrt{2}$-denominator exponent of the vector by $2$.
\end{proof}

Note also that the proof of \cref{lem:columnodd} implies that there
are no unit vectors of the form $v = u/\sqrt{2}^q$ with odd $q\geq 3$
and dimension $n<4$.

\begin{lemma}
  \label{lem:odd}
  There are no odd-dimensional unitary matrices $V=W/\sqrt{2}^q$ such
  that $W$ is an integer matrix and $q$ is odd.
\end{lemma}

\begin{proof}
  By induction on $n$. If $n=1$, then the only possibility is $V = 1$,
  hence there is no such unitary with odd least $\sqrt{2}$-denominator
  exponent. Now consider $n\geq 3$ and assume that $V=W/\sqrt{2}^q$
  where $q$ is odd and $W$ is an integer matrix.  Let $v$ be the first
  column of $V$. By \cref{lem:integraldenomexpprop}, all the
  $\sqrt{2}$-denominator exponents of $v$ are odd and by
  \cref{lem:columnodd} there exists a unitary transformation
  $G=G_1\cdots G_\ell U$ such that
  \[
  GV= \left[\begin{array}{c|ccc} 
		\frac{1}{\sqrt{2}} & & & \\ 
		\frac{1}{\sqrt{2}} & & & \\ 
		0 & & V' & \\ 
		\vdots & & & \\ 
		0 & & & 
	\end{array}\right].
  \]
  Let $u_1, u_2$ be the first two columns of $(GV)^\dagger$. Since
  $(GV)^\dagger$ is unitary, we know that $u_1^\dagger u_1 =
  u_2^\dagger u_2 = 1$, and $u_1^\dagger u_2 = u_2^\dagger u_1 =
  0$. In can then be observed from the unit condition on $u_1$ and
  $u_2$, that they each have one additional $\pm(1/\sqrt{2})$ entry,
  and are $0$ everywhere else.  Further, by the orthogonality
  condition it follows that these entries both occur on the same row
  $j$. Hence there exists $m$ such that
  \[
  (-1)_{[2]}^m X_{[2,j]} \begin{bmatrix} u_1 & u_2 \end{bmatrix}  = 
	\frac{1}{\sqrt{2}}\begin{bmatrix} 1 & 1 \\ 1 & -1 \\ 0 & 0 \\ \vdots & \vdots \\ 0 & 0 \end{bmatrix}.
  \]
  Thus
  \[
  UV((-1)_{[2]}^m X_{[2,j]})^\dagger = \left[\begin{array}{c|c} 
		H & 0 \\ \hline 0 & V''
	\end{array}\right]
  \]
  where $V''$ is a unitary matrix of the form of
  \cref{eq:integralmatrixalt} that has odd dimension and, by
  \cref{lem:integraldenomexpprop}, odd least $\sqrt{2}$-denominator
  exponent, a contradiction.
\end{proof}

Having ruled out matrices with odd dimension and odd
$\sqrt{2}$-denominator exponent, we can now prove our theorem.

\begin{theorem}
  \label{thm:int}
  If $V=W / \sqrt{2}^q$ is an $n$-dimensional unitary matrix such that
  $W$ is an integer matrix, then there exist generators $G_1, \ldots,
  G_\ell$ from \eqref{eq:intgens} such that $G_1 \cdots G_\ell V = I$.
\end{theorem}

\begin{proof}
  If $q$ is even, the result follows from \cref{thm:even}. If $q$ is
  odd, then by \cref{lem:odd} $n$ must be even, and so
  $(I_{n/2}\otimes H)V$ is a matrix with entries in $\D$. Hence the
  result follows by applying \cref{thm:even} to $(I_{n/2}\otimes H)V$.
\end{proof}

\begin{corollary}
  \label{cor:intcharac}
  A matrix $V$ can be exactly represented by an $n$-qubit circuit over
  $\s{X, CX, CCX, H}$ if and only if $V$ is a $2^n$-dimensional
  unitary matrix such that $V=W/\sqrt{2}^q$ for some integer matrix 
  $W$ and some $q\in\N$. Moreover, a single ancilla always suffices to
  construct a circuit for $V$.
\end{corollary}  

\subsection{The \texorpdfstring{$\Drtwo$}{D[√2]} case}
\label{ssec:real}

We now focus on the group of $n\times n$ unitary matrices with entries
in $\Drtwo$. The elements of this group can be written as
\begin{equation}
  \label{eq:realmatrix}
  V=\frac{1}{\sqrt{2}^q}W
\end{equation}
where $q\in\N$ and $W$ is a matrix over $\Zrtwo$. We now use
$\sqrt{2}$-denominator exponents and the relevant generators are
\begin{equation}
  \label{eq:realgens}
  \s{(-1)_{[a]},X_{[a,b]},H_{[a,b]}}
\end{equation}
where $a$ and $b$ are distinct elements of $[n]$. By
\cref{prop:realcircs}, all of the above generators can be exactly
represented by quantum circuits over the gate set $\s{X, CX, CCX, H,
  CH}$. As in the previous cases, we prove our characterization by
showing that any unitary matrix of the form \eqref{eq:realmatrix} can
be expressed as a product of generators from \eqref{eq:realgens}.

\begin{lemma}
  \label{lem:realvector}
  If $u_1,u_2\in\Zrtwo$ are such that $u_1 \equiv u_2 \pmod{2}$, then
  \[
  H \begin{bmatrix} u_1 \\ u_2 \end{bmatrix} =
  \begin{bmatrix} u_1' \\ u_2' \end{bmatrix}
  \]
  for some $u_1',u_2'\in\Zrtwo$ such that $u_1'\equiv u_2' \equiv 0
  \pmod{\sqrt{2}}$.
\end{lemma}

\begin{proof}
  Since $u_1\equiv u_2 \pmod{2}$, we have $u_1+u_2 \equiv u_1-u_2
  \equiv 0 \pmod{2}$. It can then be verified that
  \[
  H\begin{bmatrix} u_1 \\ u_2 \end{bmatrix} = \begin{bmatrix} u_1'
    \\ u_2' \end{bmatrix}
  \]
  for some $u_1'\equiv u_2'\equiv 0 \pmod{2}$.
\end{proof}

\begin{lemma}
  \label{lem:realdenommodtwo}
  If $v$ is an $n$-dimensional unit vector over $\Drtwo$ and
  $\lde_{\sqrt{2}}(v)>0$, then there exist generators $G_1, \ldots,
  G_\ell$ from \eqref{eq:realgens} such that $G_1\cdots G_\ell v = v'$
  and $\lde_{\sqrt{2}}(v') < \lde_{\sqrt{2}}(v)$.
\end{lemma}

\begin{proof}
  Write $v$ as $v=u/\sqrt{2}^q$ where $u\in\Zrtwo$ and $q>0$. Since
  $v$ is a unit vector we have $v^\dagger v=1$ and thus $2^q= \sum
  u_j^\dagger u_j = \sum u_j^2$ since $u$ is real. Letting $u_j = x_j
  + y_j \sqrt{2}$, this yields the following equation
  \[
  2^q = \sum x_j^2 + 2y_j^2 + x_jy_j2\sqrt{2}.
  \]
  Thus $\sum x_j^2 \equiv 0 \pmod{2}$ and $\sum x_jy_j=0$. It follows
  that $u_j \equiv 1 \pmod{2}$ for evenly many $j$ and $u_j\equiv
  1+\sqrt{2} \pmod{2}$ for evenly many $j$. We can therefore group
  these entries in sets of size 2 and apply \cref{lem:realvector} to
  each such set in order to reduce the $\sqrt{2}$-denominator exponent
  of the vector.
\end{proof}

The following three statements are established like the corresponding
ones in the previous section. For this reason, we omit their proofs.

\begin{lemma}
  \label{thm:realcolumn}
  Let $j\in [n]$. If $v$ is an $n$-dimensional unit vector over
  $\Drtwo$, then there exist generators $G_1, \ldots, G_\ell$ from
  \eqref{eq:realgens} such that $G_1\cdots G_\ell v = e_j$.
\end{lemma}

\begin{theorem}
  \label{thm:real}
  If $V$ is an $n$-dimensional unitary matrix with entries in
  $\Drtwo$, then there exist generators $G_1, \ldots, G_\ell$ from
  \eqref{eq:realgens} such that $G_1 \cdots G_\ell V = I$.
\end{theorem}

\begin{corollary}
  \label{cor:realcharac}
  A matrix $V$ can be exactly represented by an $n$-qubit quantum 
  circuit over the gate set $\s{X, CX, CCX, H, CH}$ if and only if 
  $V\in U_{2^n}\left( \Drtwo \right)$. Moreover, a single ancilla 
  always suffices to construct a circuit for $V$.
\end{corollary}

\subsection{The \texorpdfstring{$\Drminustwo$}{D[i√2]} case}
\label{ssec:unreal}

We now consider the group of $n\times n$ unitary matrices with entries
in $\Drminustwo$. Such matrices can be written as
\begin{equation}
  \label{eq:unrealmatrix}
  V=\frac{1}{(i\sqrt{2})^q}W
\end{equation}
where $q\in\N$ and $W$ is a matrix over $\Zrminustwo$. We now use
$i\sqrt{2}$-denominator exponents and the relevant generators are
\begin{equation}
  \label{eq:unrealgens}
  \s{(-1)_{[a]},X_{[a,b]},F_{[a,b]}}
\end{equation}
where $a$ and $b$ are distinct elements of $[n]$. By
\cref{prop:unrealcircs}, all of the above generators can be exactly
represented by quantum circuits over the gate set $\s{X, CX, CCX,
  F}$. As in the previous cases, we establish our characterization by
showing that any unitary matrix of the form \eqref{eq:unrealmatrix}
can be expressed as a product of generators from
\eqref{eq:unrealgens}.

\begin{lemma}
  \label{lem:unrealvector}
  If $u_1,u_2\in\Zrminustwo$ are such that $u_1^\dagger u_1 \equiv
  u_2^\dagger u_2 \equiv 1 \pmod{2}$, then there exist $m_0$, $m_1$,
  $m_2$, and $m_3$ such that
  \[
  F^{m_0} (-1)_{[1]}^{m_1}(-1)_{[2]}^{m_2} X^{m_3} \begin{bmatrix} u_1
    \\ u_2 \end{bmatrix} =
  \begin{bmatrix} u_1' \\ u_2' \end{bmatrix}
  \]
  for some $u_1',u_2'\in\Zrminustwo$ such that $u_1'\equiv u_2' \equiv
  0 \pmod{i\sqrt{2}}$.
\end{lemma}

\begin{proof}
  First consider the case in which $u_1\equiv u_2 \pmod{2}$. Then
  $u_1+u_2 \equiv u_1-u_2 \equiv 0 \pmod{2}$ and it can be verified
  that
  \[
  F^2\begin{bmatrix} u_1 \\ u_2 \end{bmatrix} = iH\begin{bmatrix} u_1
  \\ u_2 \end{bmatrix} = \begin{bmatrix} u_1' \\ u_2' \end{bmatrix}
  \]
  for some $u_1'\equiv u_2'\equiv 0 \pmod{i\sqrt{2}}$. We now
  consider the case in which $u_1\not\equiv u_2 \pmod{2}$. In this
  case, the fact that $u_1^\dagger u_1 \equiv u_2^\dagger u_2 \equiv 1
  \pmod{2}$ implies that one of $u_1$ or $u_2$ is congruent to $1$ or
  $3$ modulo $2i\sqrt{2}$ while the other is congruent to $(1
  +i\sqrt{2})$ or $(3 +i\sqrt{2})$ modulo
  $2i\sqrt{2}$. We can therefore find $m_1,m_2,m_3$ such that
  \[
  (-1)^{m_1}_{[1]}(-1)_{[2]}^{m_2} X^{m_3} \begin{bmatrix} u_1
    \\ u_2 \end{bmatrix} =
  \begin{bmatrix} u_1'' \\ u_2'' \end{bmatrix}
  \]
  where $u_1'' \equiv 1+i\sqrt{2} \pmod{2i\sqrt{2}}$ and
  $u_2'' \equiv 1 \pmod{2i\sqrt{2}}$. Then
  \[
  F \begin{bmatrix} u_1'' \\ u_2'' \end{bmatrix} = \frac{1}{2}
  \begin{bmatrix} (1+i\sqrt{2})u_1'' + u_2'' \\
    u_1'' + (-1+i\sqrt{2})u_2'' \end{bmatrix}.
  \]
  But $u_1'' \equiv 1+i\sqrt{2} \pmod{2i\sqrt{2}}$ and
  $u_2'' \equiv 1 \pmod{2i\sqrt{2}}$ so that
  \[
  (1+i\sqrt{2})u_1'' + u_2'' \equiv (1+i\sqrt{2})^2 + 1
  \equiv 2i\sqrt{2} \equiv 0 \pmod{2i\sqrt{2}}.
  \]
  and
  \[
  u_1'' + (-1+i\sqrt{2})u_2'' \equiv (1+i\sqrt{2}) + (-1+i\sqrt{2})
  \equiv 2i\sqrt{2} \equiv 0 \pmod{2i\sqrt{2}}.
  \]
  Hence we can set $u_1' = ((1+i\sqrt{2})u_1'' + u_2'')/2$ and
  $u_2' = ( u_1'' + (-1+i\sqrt{2})u_2'')/2$ to complete the
  proof.
\end{proof}

\begin{lemma}
  \label{lem:unrealdenommodtwo}
  If $v$ is an $n$-dimensional unit vector over $\Drminustwo$ and
  $\lde_{i\sqrt{2}}(v)>0$, then there exist generators $G_1,
  \ldots, G_\ell$ from \eqref{eq:unrealgens} such that $G_1\cdots
  G_\ell v = v'$ and $\lde_{i\sqrt{2}}(v') <
  \lde_{i\sqrt{2}}(v)$.
\end{lemma}

\begin{proof}
  Write $v$ as $v=u/i\sqrt{2}^q$ where $u\in\Zrminustwo$ and
  $q>0$. Since $v$ is a unit vector we have $v^\dagger v=1$ and thus
  $(\mbox{-}2)^q= \sum u_j^\dagger u_j$. Thus $\sum u_j^\dagger u_j
  \equiv 0 \pmod{2}$ and it follows that $u_j^\dagger u_j \equiv 1
  \pmod{2}$ for evenly many $j$, since modulo 2 we have $u_j^\dagger
  u_j \equiv 0$ or $u_j^\dagger u_j \equiv 1$. We can therefore group
  these entries in sets of size 2 and apply \cref{lem:unrealvector} to
  each such set in order to reduce the denominator exponent.
\end{proof}

\begin{lemma}
  \label{thm:unrealcolumn}
  Let $j\in [n]$. If $v$ is an $n$-dimensional unit vector over
  $\Drminustwo$, then there exist generators $G_1, \ldots, G_\ell$
  from \eqref{eq:unrealgens} such that $G_1\cdots G_\ell v = e_j$.
\end{lemma}

\begin{theorem}
  \label{thm:unreal}
  If $V$ is an $n$-dimensional unitary matrix with entries in
  $\Drminustwo$, then there exist generators $G_1, \ldots, G_\ell$
  from \eqref{eq:unrealgens} such that $G_1 \cdots G_\ell V = I$.
\end{theorem}

\begin{corollary}
  \label{cor:unrealcharac}
  A matrix $V$ can be exactly represented by an $n$-qubit circuit over
  $\s{X, CX, CCX, F}$ if and only if $V\in U_{2^n}\left( \Drminustwo
  \right)$. Moreover, a single ancilla always suffices to construct a
  circuit for $V$.
\end{corollary}  

We close this section with a characterization of ancilla-free circuits
over $\s{ X, CX, CCX, F}$, focusing on circuits on four or more
qubits. The required circuit constructions are relegated to
\cref{app:unreal}.

\begin{corollary}
  \label{prop:ancillafreeunreal}
  Let $n\geq 4$. A matrix $V\in U_{2^n}(\Drminustwo)$ can be exactly
  represented by an ancilla-free $n$-qubit circuit over $\s{ X, CX,
    CCX, F}$ if and only if $\det{V} = 1$.
\end{corollary}

\begin{proof}
  If $\det{V} \neq 1$, then $V$ cannot be exactly represented over
  $\s{X, CX, CCX, F}$ without ancillas when $n\geq 4$, as each gate
  has determinant $1$ in this case.
	
  Now suppose $\det{V} = 1$.  First observe that in
  \cref{lem:unrealvector}, and consequently
  \cref{lem:unrealdenommodtwo}, the least
  $i\sqrt{2}$-denominator exponent can be reduced by
  substituting $F^{m_0}(-1)_{[1]}^{m_1}(-1)_{[2]}^{m_2}X^{m_3}$ as
  follows:
  \begin{align*}
	F^2 &\rightarrow (FZ)(ZF), \\
	F(-1)_{[1]} &\rightarrow (FZ)(XZ)(XZ), \\
	F(-1)_{[2]} &\rightarrow (FZ), \\
	F(-1)_{[1]}(-1)_{[2]} &\rightarrow (ZF)(XZ)(XZ), \\
	FX &\rightarrow (FZ)(ZX), \\
	F(-1)_{[1]}X &\rightarrow (ZF)(XZ), \\
        F(-1)_{[2]}X &\rightarrow (ZF)(ZX), \mbox{ and } \\
	F(-1)_{[1]}(-1)_{[2]}X &\rightarrow (FZ)(XZ).
  \end{align*}
  Moreover, $FZ$, $ZF$, $XZ$ and $ZX$ have determinant $1$ and
  can be represented over $\s{ X, CX, CCX, F}$ without
  ancillas when $n \geq 4$ as shown in \cref{app:unreal}. Likewise, an
  analogue of \cref{thm:unrealcolumn}, where $G_1\cdots G_\ell v =
  i^me_j$, holds by using the two-level $ZX$ operator, which has
  determinant $1$ and is representable without ancillas as shown in
  \cref{app:unreal}. It now suffices to note that there exist
  generators $G_1,\dots,G_\ell$ from the set $\s{XZ_{[a,b]},
    ZX_{[a,b]}, FZ_{[a,b]}, ZF_{[a,b]}}$ such that \[ G_1\cdots G_\ell
  V = D \] where $D$ is a diagonal unitary with entries $\pm 1$.
  Since $\det D = \det V = 1$, there are an even number of $-1$
  entries, thus we can group them into pairs and use four-level
  $I\otimes Z$ operators to change each pair into $+1$ and obtain the
  identity matrix.
\end{proof}

\subsection{The \texorpdfstring{$\Di$}{D[i]} case}
\label{ssec:gaussian}

Finally, we turn our attention to the group of $n\times n$ unitary
matrices with entries in $\Di$. The relevant set of generators is
\begin{equation}
  \label{eq:evengaussiangens}
  \s{i_{[a]},X_{[a,b]},\omega H_{[a,b]}}
\end{equation}
where $a$ and $b$ are distinct elements of $[n]$. We reason as in the
previous cases, noting by \cref{prop:gaussiancircs} that all of the
above generators can be exactly represented by circuits over
$\s{X, CX, CCX, \omega H, S}$.

If $V$ is a matrix over $\Di$, then $V$ can be written as $V=W/2^q$
where $q\in\N$ and $W$ is a matrix over $\Zi$. For our purposes,
however, it is more convenient to express these matrices as
\begin{equation}
  \label{eq:gaussianmatrixeven}
  V=\frac{1}{(1+i)^q}W
\end{equation}
where $q\in\N$ and $W$ is a matrix over $\Zi$. This is equivalent
since
\[
  \frac{1}{2^q}W = \frac{i^q}{(1+i)^{2q}}W = \frac{1}{(1+i)^{2q}}W'.
\]
We therefore use matrices of the form \eqref{eq:gaussianmatrixeven}
and use $(1+i)$-denominator exponents.

\begin{lemma}
  \label{lem:gaussianvector}
  If $u_1,u_2\in\Zi$ are such that $u_1^2 \equiv u_2^2 \equiv 1
  \pmod{2}$, then there exist $m_1$ and $m_2$ such that
  \[
  \omega H i_{[1]}^{m_1}i_{[2]}^{m_2}\begin{bmatrix} u_1
    \\ u_2 \end{bmatrix} =
  \begin{bmatrix} u_1' \\ u_2' \end{bmatrix}
  \]
  for some $u_1',u_2'\in\Zi$ such that $u_1'\equiv u_2' \equiv 0
  \pmod{1+i}$.
\end{lemma}

\begin{proof}
  If $u^2\equiv 1 \pmod{2}$, then $u\equiv 1 \pmod{2}$ or $u\equiv i
  \pmod{2}$. Furthermore, if $u\equiv i \pmod{2}$, then $iu\equiv 1
  \pmod{2}$. Hence, given $u_1,u_2\in\Z$ such that $u_1^2 \equiv u_2^2
  \equiv 1 \pmod{2}$, we can find $m_1$ and $m_2$ such that
  $i^{m_1}u_1 \equiv i^{m_2}u_2 \equiv 1 \pmod{2}$. It can then be
  verified that
  \[
  \omega H i_{[1]}^{m_1}i_{[2]}^{m_2}\begin{bmatrix} u_1
    \\ u_2 \end{bmatrix} =
  \begin{bmatrix} u_1' \\ u_2' \end{bmatrix}
  \]
  for some $u_1'\equiv u_2'\equiv 0 \pmod{1+i}$.
\end{proof}

\begin{lemma}
  \label{lem:gaussiandenommodtwo}
  If $v$ is an $n$-dimensional unit vector over $\Di$ and
  $\lde_{(1+i)}(v)>0$, then there exist generators $G_1, \ldots,
  G_\ell$ from \eqref{eq:evengaussiangens} such that $G_1\cdots G_\ell
  v = v'$ and $\lde_{(1+i)}(v') < \lde_{(1+i)}(v)$.
\end{lemma}

\begin{proof}
  Write $v$ as $v=u/(1+i)^q$ where $u\in\Zi$ and $q>1$. Since
  $(1+i)^\dagger (1+i) = 2$ and $v$ is a unit vector, we have $2^q=
  \sum u_j^\dagger u_j$. Thus $ 0 \equiv \sum u_j^\dagger u_j \equiv
  \sum u_j^2 \pmod{2}$ and it follows that $u_j^2 \equiv 1 \pmod{2}$
  for evenly many $j$. We can therefore group these entries in sets of
  size 2 and apply \cref{lem:gaussianvector} to each such set in order
  to reduce the denominator exponent.
\end{proof}

\begin{lemma}
  \label{lem:gaussiancolumnevendenom}
  Let $j\in [n]$. If $v$ is an $n$-dimensional unit vector over $\Di$,
  then there exist generators $G_1, \ldots, G_\ell$ from
  \eqref{eq:evengaussiangens} such that $G_1\cdots G_\ell v = e_j$.
\end{lemma}

\begin{theorem}
  \label{thm:gaussianeven}
  If $V$ is an $n$-dimensional unitary matrix with entries in $\Di$,
  then there exist generators $G_1, \ldots, G_\ell$ from
  \eqref{eq:evengaussiangens} such that $G_1 \cdots G_\ell V = I$.
\end{theorem}

\begin{corollary}
  \label{cor:gaussianevencharac}
  A matrix $V$ can be exactly represented by an $n$-qubit circuit over
  $\s{X, CX, CCX, \omega H, S}$ if and only if $V\in
  U_{2^n}(\Di)$. Moreover, a single ancilla always suffices to
  construct a circuit for $V$.
\end{corollary}

\begin{corollary}
	\label{prop:ancillafreegauss}
	Let $n\geq 4$. A matrix $V\in U_{2^n}(\Di)$ can be exactly represented by an 
	ancilla-free $n$-qubit circuit over $\s{ X, CX, CCX, \omega H, S}$ 
	if and only if $\det{V} = 1$.
\end{corollary}

\begin{proof}
  We proceed as in the proof of \cref{prop:ancillafreeunreal}.  In
  particular, the least $(1+i)$-denominator exponent can be reduced in
  \cref{lem:gaussiandenommodtwo} by substituting $\omega
  Hi_{[1]}^{m_1}i_{[2]}^{m_2}$ with ancilla-free two-level generators
  as follows:
  \begin{align*}
  	\omega H &\rightarrow (\omega S H), \\
	\omega Hi_{[1]} &\rightarrow (\omega H S)(iZ), \\
	\omega Hi_{[2]} &\rightarrow (\omega H S), \mbox{ and } \\
	\omega Hi_{[1]}i_{[2]} &\rightarrow (\omega SH)(iZ).
  \end{align*}
  Moreover, each parenthesized two-level operator on the right hand
  side has determinant $1$ and can be exactly represented over $\s{ X,
    CX, CCX, \omega H, S}$ without ancillas when $n \geq 4$ as shown
  in \cref{app:gaussian}.
	
  An analogue of \cref{lem:gaussiancolumnevendenom} where $G_1\cdots
  G_\ell v = i^me_j$ holds by using the two-level $iX$ operator, which
  similarly has determinant $1$ and is representable without ancillas
  as shown in \cref{app:gaussian}. Again, there exist generators
  $G_1,\dots,G_\ell$ from the set $\s{iZ_{[a,b]}, iX_{[a,b]}, \omega
    SH_{[a,b]}, \omega HS_{[a,b]}}$ such that \[ G_1\cdots G_\ell V =
  D \] where $D$ is a diagonal unitary with entries $i^m$.  We can
  then use the $n$-qubit two-level $i Z$ operator to remove the phases
  as follows. Suppose the $j$th diagonal entry is $i^{m_j}$ and let
  $N=2^n$. It can then be observed that
  \[
  (iZ_{[1,2]})^{-m_1}(iZ_{[2,3]})^{-m_1-m_2}
  \dots (iZ_{[N-1,N]})^{-m_1-m_2-\dots-m_{N-1}}D =  
  \begin{bmatrix}
  	I_{N-1} & 0 \\
	0 & i^{\sum_{j=1}^N m_j}
	\end{bmatrix}= \begin{bmatrix}
	I_{N-1} & 0 \\
	0 & \det D
	\end{bmatrix}
  \]
  Since $\det{D} = \det{V} = 1$, the proof is complete.
\end{proof}

\subsubsection{Super-Gaussian Clifford+\texorpdfstring{$T$}{T} operators}

As in the integral case, the characterization of Gaussian Clifford+$T$
circuits as unitaries over $\Di$ requires the unusual $\omega H$ gate
as a generator. Replacing $\omega H$ with $H$ yields a slightly larger
set of unitaries with matrices of the form
\begin{equation}
  \label{eq:gaussianmatrixalt}
  V=\frac{1}{\sqrt{2}^q}W
\end{equation}
where $q\in\N$ and $W$ is a matrix over $\Zi$.

We use \cref{cor:gaussianevencharac} together with
\cref{cor:gaussiancircs} to show that any unitary of the form of
\cref{eq:gaussianmatrixalt} can be represented by a circuit over the
gate set $\s{X, CX, CCX, H, S}$. In this case we use
$\sqrt{2}$-denominator exponents and, as in \cref{ssec:integral}, we
make use of the fact that $\sqrt{2}\notin\Zi$.  The relevant
generators are now
\begin{equation}
  \label{eq:gaussiangens}
  \s{i_{[a]},X_{[a,b]},\omega H_{[a,b]}, \omega I_n}.
\end{equation}

\begin{lemma}
  \label{lem:gaussiandenomexpprop}
  If $V\neq 0$ is as in \eqref{eq:gaussianmatrixalt}, then all the
  denominator exponents of $V$ are congruent modulo 2.
\end{lemma}

\begin{proof}
  Similar to the proof of \cref{lem:integraldenomexpprop}.
\end{proof}

\begin{theorem}
  \label{thm:gaussian}
  If $V=W / \sqrt{2}^q$ is an $n$-dimensional unitary matrix such that
  $W$ is a matrix over $\Zi$, then there exist generators $G_1,
  \ldots, G_\ell$ from \eqref{eq:gaussiangens} such that $G_1 \cdots
  G_\ell V = I$.
\end{theorem}

\begin{proof}
  If $q$ is even, the result follows from \cref{thm:gaussianeven}. If
  $q$ is odd, then $(\omega I_n)V$ is a matrix with entries in
  $\Di$. Hence the result follows by applying \cref{thm:gaussianeven}
  to $(\omega I_n)V$.
\end{proof}

\begin{corollary}
  \label{cor:gaussiancharac}
  A matrix $V$ can be exactly represented by an $n$-qubit circuit over
  $\s{X, CX, CCX, H, S}$ if and only if $V$ is a $2^n$-dimensional
  unitary matrix such $V=W/\sqrt{2}^q$ for some matrix $W$ over $\Zi$
  and some $q\in\N$. Moreover, a single ancilla always suffices to
  construct a circuit for $V$.
\end{corollary}

\section{Conclusion}
\label{sec:conc}

In this paper, we provided number-theoretic characterizations for
several classes of restricted but universal Clifford+$T$ circuits,
focusing on integral, real, imaginary, and Gaussian circuits. We
showed that a unitary matrix can be exactly represented by an
$n$-qubit integral Clifford+$T$ circuit if and only if it is an
element of the group $U_{2^n}(\D)$. We then established that real,
imaginary, and Gaussian circuits similarly correspond to the groups
$U_{2^n}(\Drtwo)$, $U_{2^n}(\Drminustwo)$, and $U_{2^n}(\Di)$,
respectively.

An avenue for future research is to improve the performance, in
runtime or gate count, of the algorithms introduced in the present
paper. Further afield, it would be interesting to study restricted
Clifford+$T$ circuits in the context of fault-tolerance, randomized
benchmarking, or simulation. While these and many other questions
remain open, we hope that our characterizations will help deepen our
understanding of Clifford+$T$ circuits, restricted or not.

\section{Acknowledgements}
\label{sec:acknowledgements}

We would like to thank Alexandre Cl\'ement, Sarah Li, Rob Noble, and
Kira Scheibelhut for their valuable insights as well as Xiaoning Bian,
Simon Burton, and Peter Selinger for providing helpful feedback on an
earlier version of the paper. We are especially indebted to Xiaoning
Bian for the proof of \cref{lem:odd} and to Peter Selinger for crucial
observations on restrictions of the Clifford+$T$ gate set.

\bibliographystyle{abbrvnat} 
\bibliography{restricted}

\appendix

\section{Ancilla-Free Circuit Constructions}

\subsection{The \texorpdfstring{$\Drminustwo$}{D[√-2]} case}\label{app:unreal}

In this appendix, we give ancilla-free constructions of the two-level
operators $XZ$, $ZX=(XZ)^\dagger$, $FZ$, and $ZF$ over $\s{ X, CX,
  CCX, F}$. We progressively build up to the necessary operators.

\begin{lemma}
	For any $n$, the $n$-qubit two-level $X$ and $Z$ gates can be represented 
	over the gate set $\s{X, CX, CCX, F}$ with a single dirty ancilla.
\end{lemma}

\begin{proof}
	Recall that the two-level $X$ gate is representable over
        $\s{X, CX, CCX}$ with a single dirty ancilla \cite{BBC95}.
        The two-level $Z$ gate can then be constructed as follows.
	
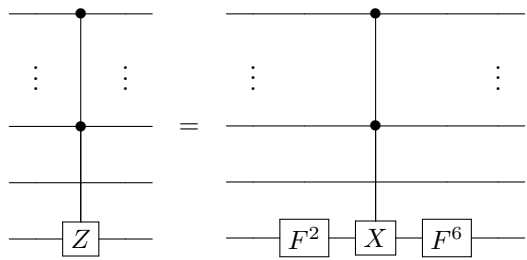
\[
	\Qcircuit @C=1em @R=.84em @!R {
		& \qw & \ctrl{4} & \qw & \qw \\
		& \vdots  & & \vdots & \\
		& \qw & \ctrl{2} & \qw & \qw  \\
		& \qw & \qw & \qw & \qw \\
		& \qw & \gate{Z} & \qw & \qw
	}
	\raisebox{-4.25em}{\quad=\quad}
	\Qcircuit @C=1em @R=.7em @!R {
		& \qw & \qw & \ctrl{4} & \qw & \qw & \qw \\
		& \vdots  & & & & \vdots & \\
		& \qw& \qw & \ctrl{2} & \qw & \qw & \qw \\
		& \qw & \qw & \qw & \qw & \qw & \qw \\
		& \qw & \gate{F^2} & \gate{X} & \gate{F^6} & \qw & \qw
	}
	\]
\end{proof}

\begin{lemma}
	For any $n$, the $n$-qubit two-level $ZXF$ gate can be represented over  
	$\s{X, CX, CCX, F}$ with a single dirty ancilla.
\end{lemma}

\begin{proof}
	Recall that
	\begin{align*}
		(ZXF)^2 &= I \\
		X(ZXF)X(ZXF)X &= ZXF.
	\end{align*}
	Hence it follows that the two-level $ZXF$ gate can be
        implemented over $\s{X, CX, CCX, F}$ with a single dirty
        ancilla as follows.
	
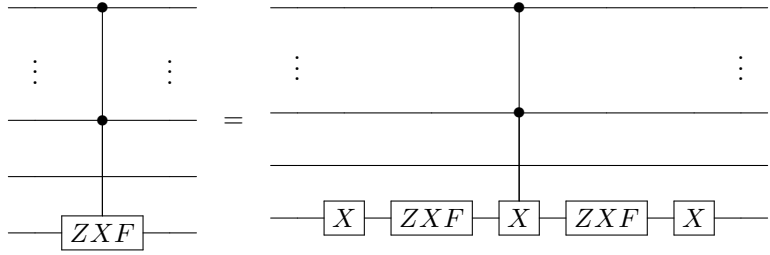
\[
	\Qcircuit @C=1em @R=.84em @!R {
		& \qw & \ctrl{4} & \qw & \qw \\
		& \vdots  & & \vdots & \\
		& \qw & \ctrl{2} & \qw & \qw  \\
		& \qw & \qw & \qw & \qw \\
		& \qw & \gate{ZXF} & \qw & \qw
	}
	\raisebox{-4.25em}{\quad=\quad}
	\Qcircuit @C=1em @R=.7em @!R {
		& \qw & \qw & \qw & \ctrl{4} & \qw & \qw & \qw & \qw \\
		& \vdots & & & & & & \vdots & \\
		& \qw & \qw & \qw & \ctrl{2} & \qw & \qw & \qw & \qw \\
		& \qw & \qw & \qw & \qw & \qw & \qw & \qw & \qw \\
		& \qw & \gate{X} & \gate{ZXF} & \gate{X} & \gate{ZXF} & \gate{X} & \qw & \qw
	}
	\]
\end{proof}

\begin{proposition}
	\label{prop:qzone}
	For any $n$, the $n$-qubit two-level $XZ$ gates can be represented 
	without ancillas over $\{X, CX, CCX, F\}$.
\end{proposition}

\begin{proof}
	Recall that the controlled-$F$ gate is representable with a
        single dirty ancilla.  The two-level $XZ$ gate is then
        constructible without ancillas using the following circuit.
	
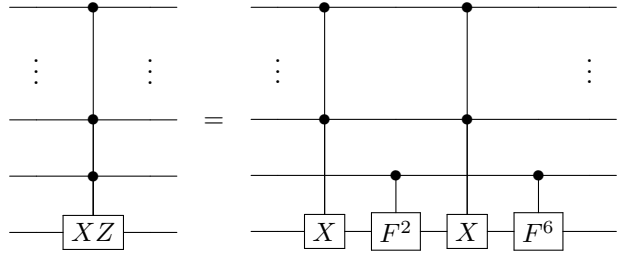
\[
	\Qcircuit @C=1em @R=.84em @!R {
		& \qw & \ctrl{4} & \qw & \qw \\
		& \vdots  & & \vdots & \\
		& \qw & \ctrl{2} & \qw & \qw  \\
		& \qw & \ctrl{1} & \qw & \qw \\
		& \qw & \gate{XZ} & \qw & \qw
	}
	\raisebox{-4.25em}{\quad=\quad}
	\Qcircuit @C=1em @R=.7em @!R {
		& \qw & \ctrl{4} & \qw & \ctrl{4} & \qw & \qw & \qw \\
		& \vdots & & & & & \vdots & \\
		& \qw & \ctrl{2} & \qw & \ctrl{2} & \qw & \qw & \qw \\
		& \qw & \qw & \ctrl{1} & \qw & \ctrl{1} & \qw & \qw \\
		& \qw & \gate{X} & \gate{F^2} & \gate{X} & \gate{F^6} & \qw & \qw
	}
	\]
\end{proof}

\begin{proposition}
	\label{prop:qztwo}
	For any $n$, the $n$-qubit two-level $FZ$ and $ZF$ gates can be 
	represented without ancillas over $\{X, CX, CCX, F\}$.
\end{proposition}

\begin{proof}
	In the $n=1$ and $n=2$ cases, this is trivially true, as the
        controlled $F$ and controlled $Z$ are both implementable
        without ancillas. For $n\geq3$, note that $XZ=-ZX$, and so
        $(ZXF)X(ZXF)X = X(ZXF) = -ZF$. Hence, we have the equality below.
	
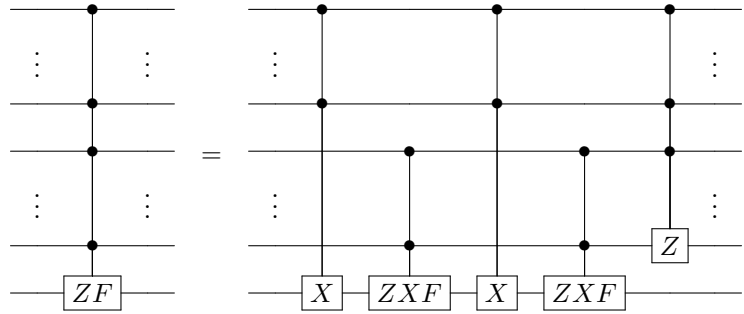
\[
	\Qcircuit @C=1em @R=.5em @!R {
		& \qw & \ctrl{6} & \qw & \qw  \\
		& \vdots & & \vdots & \\
		&  \qw & \ctrl{4} & \qw & \qw \\
		& \qw & \ctrl{3} & \qw & \qw  \\
		& \vdots & & \vdots & \\
		&  \qw & \ctrl{1} & \qw & \qw \\
		& \qw & \gate{ZF} & \qw & \qw
	}
	\raisebox{-5.5em}{\quad=\quad}
	\Qcircuit @C=1em @R=.5em @!R {
		& \qw & \ctrl{6} & \qw  & \ctrl{6} & \qw & \ctrl{5} & \qw & \qw\\
		& \vdots & & & & & & \vdots & \\
		& \qw & \ctrl{4} & \qw  & \ctrl{4} & \qw & \ctrl{3} & \qw & \qw\\
		& \qw & \qw & \ctrl{3}  & \qw & \ctrl{3} & \ctrl{2} & \qw & \qw\\
		& \vdots & & & & & & \vdots & \\
		& \qw & \qw & \ctrl{1}  & \qw & \ctrl{1} & \gate{Z} & \qw & \qw\\
		& \qw & \gate{X} & \gate{ZXF} & \gate{X} & \gate{ZXF} & \qw & \qw & \qw
	}
	\]
	Finally for $FZ$, we can note that
	\begin{align*}
		F^6(XZ) X(ZXF)X(ZXF)F^2 
			&= F^6(XZ)(ZXF)XF^2 \\
			&= FF^6 X F^2 \\
			&=FZ,
	\end{align*}
	giving the following circuit identity
	
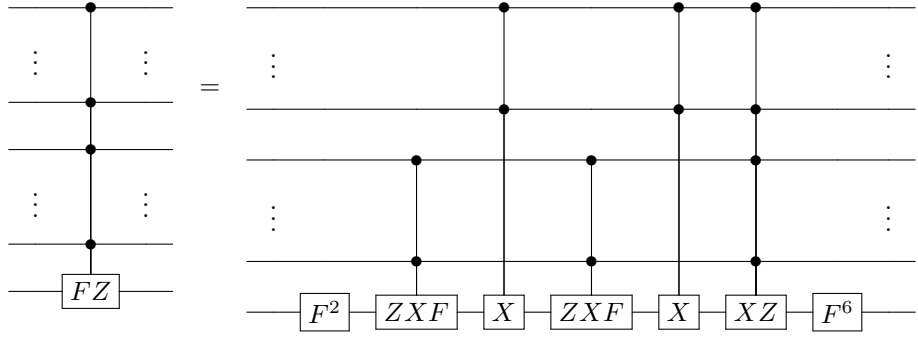
\[
	\Qcircuit @C=1em @R=.5em @!R {
		& \qw & \ctrl{6} & \qw & \qw  \\
		& \vdots & & \vdots & \\
		&  \qw & \ctrl{4} & \qw & \qw \\
		& \qw & \ctrl{3} & \qw & \qw  \\
		& \vdots & & \vdots & \\
		&  \qw & \ctrl{1} & \qw & \qw \\
		& \qw & \gate{FZ} & \qw & \qw
	}
	\raisebox{-3em}{\quad=\quad}
	\Qcircuit @C=1em @R=.5em @!R {
		& \qw & \qw & \qw & \ctrl{6} & \qw & \ctrl{6} & \ctrl{6} & \qw & \qw & \qw \\
		& \vdots & & & & & & & & \vdots & \\
		& \qw & \qw & \qw & \ctrl{4} & \qw & \ctrl{4} & \ctrl{4} & \qw & \qw & \qw\\
		& \qw & \qw & \ctrl{3} & \qw & \ctrl{3} & \qw & \ctrl{3} & \qw & \qw  & \qw\\
		& \vdots & & & & & & & & \vdots & \\
		& \qw & \qw & \ctrl{1} & \qw & \ctrl{1} & \qw & \ctrl{1} & \qw & \qw & \qw\\
		& \qw & \gate{F^2} & \gate{ZXF} & \gate{X} & \gate{ZXF} & \gate{X} & \gate{XZ} & \gate{F^6} & \qw & \qw
	}
	\]
	where no ancillas are needed.
\end{proof}

\subsection{The \texorpdfstring{$\Di$}{D[i]} case}
\label{app:gaussian}

In this appendix we give ancilla-free constructions of the two-level
operators $iX$, $iZ$, $\omega SH$, and $\omega H S$ over $\s{ X, CX,
  CCX, \omega H, S}$. We progressively build up to the necessary
operators.

\begin{lemma}
	For any $n$, the $n$-qubit two-level $X$ and $Z$ gates can be 
	represented over the gate set $\s{X, CX, CCX, \omega H, S}$ with a single
    dirty ancilla.
\end{lemma}

\begin{proof}
	Recall that the two-level $X$ gate is representable over
        $\s{X, CX, CCX}$ with a single dirty ancilla \cite{BBC95}.
        The two-level $Z$ gate can then be constructed as follows.
	
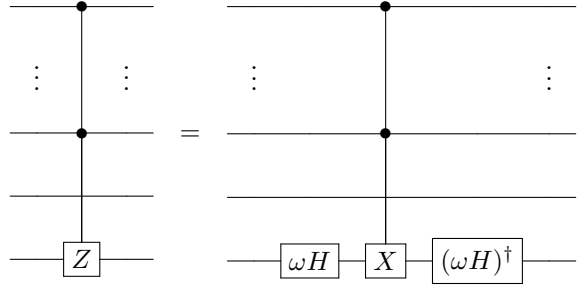
\[
	\Qcircuit @C=1em @R=1.1em @!R {
		& \qw & \ctrl{4} & \qw & \qw \\
		& \vdots  & & \vdots & \\
		& \qw & \ctrl{2} & \qw & \qw  \\
		& \qw & \qw & \qw & \qw \\
		& \qw & \gate{Z} & \qw & \qw
	}
	\raisebox{-4.8em}{\quad=\quad}
	\Qcircuit @C=1em @R=.7em @!R {
		& \qw & \qw & \ctrl{4} & \qw & \qw & \qw \\
		& \vdots  & & & & \vdots & \\
		& \qw& \qw & \ctrl{2} & \qw & \qw & \qw \\
		& \qw & \qw & \qw & \qw & \qw & \qw \\
		& \qw & \gate{\omega H} & \gate{X} & \gate{(\omega H)^\dagger} & \qw & \qw
	}
	\]
\end{proof}

\begin{lemma}
	For any $n$, the $n$-qubit two-level $S$ operators can be 
	represented over $\s{X, CX, CCX, \omega H, S}$ with two
    dirty ancillas.
\end{lemma}

\begin{proof}
	Observe that the two-level $S$ operator can be constructed as
        follows.
	
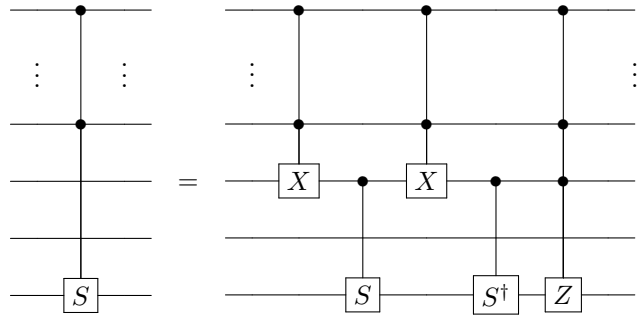
\[
	\Qcircuit @C=1em @R=.87em @!R {
		& \qw & \ctrl{5} & \qw & \qw \\
		& \vdots  & & \vdots & \\
		& \qw & \ctrl{3} & \qw & \qw  \\
		& \qw & \qw & \qw & \qw \\
		& \qw & \qw & \qw & \qw \\
		& \qw & \gate{S} & \qw & \qw
	}
	\raisebox{-6.5em}{\quad=\quad}
	\Qcircuit @C=1em @R=.7em @!R {
		& \qw & \ctrl{3} & \qw & \ctrl{3} & \qw & \ctrl{5} & \qw & \qw \\
		& \vdots  & & & & & & & \vdots & \\
		& \qw & \ctrl{1} & \qw & \ctrl{1} & \qw & \ctrl{3} & \qw & \qw \\
		& \qw & \gate{X}  & \ctrl{2} & \gate{X} & \ctrl{2} & \ctrl{2} & \qw & \qw \\
		& \qw & \qw & \qw & \qw & \qw & \qw & \qw & \qw \\
		& \qw & \qw & \gate{S} & \qw & \gate{S^\dagger} & \gate{Z} & \qw & \qw
	}
	\]
\end{proof}

\begin{proposition}
	For any $n$, the $n$-qubit two-level $iX$ operators can be represented 
	without ancillas over $\{X, CX, CCX, \omega H, S\}$.
\end{proposition}

\begin{proof}
	For $n=1$, we have
	\[iX = (\omega H)S^2(\omega H).\]
	For $n=2$, we have
	
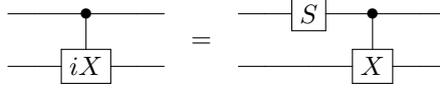
\[
	\Qcircuit @C=1em @R=.7em @!R {
		& \qw & \ctrl{1} & \qw & \qw \\
		& \qw & \gate{iX} & \qw & \qw
	}
	\raisebox{-1em}{\quad=\quad}
	\Qcircuit @C=1em @R=.7em @!R {
		& \qw & \gate{S} & \ctrl{1} & \qw & \qw\\
		& \qw & \qw & \gate{X} &  \qw & \qw
	}\]
	and for $n\geq 3$, we have the circuit below.
	
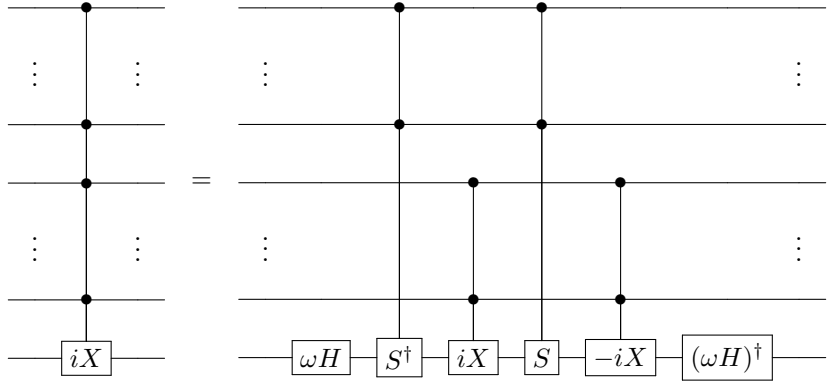
\[
	\Qcircuit @C=1em @R=.92em @!R {
		& \qw & \ctrl{6} & \qw & \qw \\
		& \vdots  & & \vdots & \\
		& \qw & \ctrl{4} & \qw & \qw  \\
		& \qw & \ctrl{3} & \qw & \qw  \\
		& \vdots & & \vdots & \\
		&  \qw & \ctrl{1} & \qw & \qw \\
		& \qw & \gate{iX} & \qw & \qw
	}
	\raisebox{-6.5em}{\quad=\quad}
	\Qcircuit @C=1em @R=.5em @!R {
		& \qw& \qw & \ctrl{6} & \qw & \ctrl{6} & \qw & \qw & \qw  & \qw\\
		& \vdots & & & & & & & \vdots & \\
		& \qw & \qw & \ctrl{4} & \qw & \ctrl{4} & \qw & \qw & \qw & \qw\\
		& \qw& \qw & \qw  & \ctrl{3} & \qw & \ctrl{3} & \qw & \qw  & \qw\\
		& \vdots & & & & & & & \vdots & \\
		& \qw& \qw & \qw  & \ctrl{1} & \qw & \ctrl{1} & \qw & \qw & \qw\\
		& \qw & \gate{\omega H} & \gate{S^\dagger} & \gate{iX} & \gate{S} & \gate{-iX} & \gate{(\omega H)^\dagger} & \qw & \qw
	}
	\]

\end{proof}

\begin{proposition}
	For any $n$, the $n$-qubit two-level $iZ$ operator can be 
	represented without ancillas over $\{X, CX, CCX, \omega H, S\}$.
\end{proposition}

\begin{proof}
	Using the $n$-qubit $iX$ operator as follows.
	
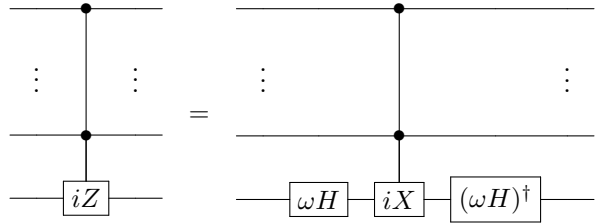
\[
	\Qcircuit @C=1em @R=1.1em @!R {
		& \qw & \ctrl{3} & \qw & \qw \\
		& \vdots  & & \vdots & \\
		& \qw & \ctrl{1} & \qw & \qw  \\
		& \qw & \gate{iZ} & \qw & \qw
	}
	\raisebox{-4em}{\quad=\quad}
	\Qcircuit @C=1em @R=.7em @!R {
		& \qw & \qw & \ctrl{3} & \qw & \qw & \qw \\
		& \vdots  & & & & \vdots & \\
		& \qw& \qw & \ctrl{1} & \qw & \qw & \qw \\
		& \qw & \gate{\omega H} & \gate{iX} & \gate{(\omega H)^\dagger} & \qw & \qw
	}
	\]
\end{proof}

\begin{proposition}
	For any $n\neq 2$, the $n$-qubit two-level $\omega SH$ and
    $\omega HS$ operators are can be represented without ancillas 
    over $\{X, CX, CCX,\omega H, S\}$. For $n=2$, the $2$-qubit 
    two-level $\omega SH$ and $\omega HS$ operators can be 
    represented if $CS$ is appended to the gate list.
\end{proposition}

\begin{proof}
	The $n=1$ case is trivially true. For $n=2$, if we do not use
        ancillas, we are left with the Clifford group generators and
        are unable to implement non-Clifford operators. However,
        appending $CS$ to our list of operators gives us the following
        circuits.
	
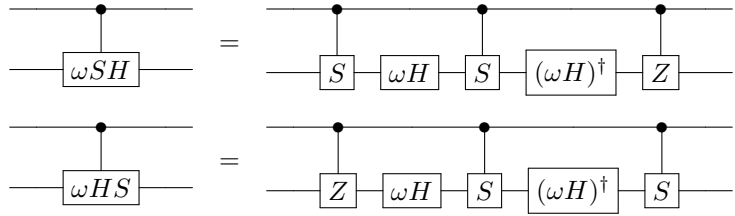
\[
	\Qcircuit @C=1em @R=1em @!R {
		& \qw & \ctrl{1} & \qw & \qw \\
		& \qw & \gate{\omega S H} & \qw & \qw
	}
	\raisebox{-1.3em}{\quad=\quad}
	\Qcircuit @C=1em @R=.7em @!R {
		& \qw & \ctrl{1} & \qw & \ctrl{1} & \qw & \ctrl{1} & \qw & \qw\\
		& \qw & \gate{S} & \gate{\omega H} & \gate{S} & \gate{(\omega H)^\dagger} & \gate{Z} & \qw & \qw
	}\]
	\[
	\Qcircuit @C=1em @R=1em @!R {
		& \qw & \ctrl{1} & \qw & \qw \\
		& \qw & \gate{\omega H S} & \qw & \qw
	}
	\raisebox{-1.3em}{\quad=\quad}
	\Qcircuit @C=1em @R=.7em @!R {
		& \qw & \ctrl{1} & \qw & \ctrl{1} & \qw & \ctrl{1} & \qw & \qw\\
		& \qw & \gate{Z} & \gate{\omega H} & \gate{S} & \gate{(\omega H)^\dagger} & \gate{S} & \qw & \qw
	}\]
	Note moreover that the $2$-qubit circuits can also be represented
	over $\s{X, CX, CCX, \omega H, S}$ with a single dirty ancilla,
	which we use in the identities for $n\geq3$ below.
	\[
	\Qcircuit @C=1em @R=1.1em @!R {
		& \qw & \ctrl{3} & \qw & \qw \\
		& \vdots & & \vdots& \\
		& \qw & \ctrl{2} & \qw & \qw  \\
		& \qw & \ctrl{1} & \qw & \qw \\
		& \qw & \gate{\omega S H} & \qw & \qw
	}
	\raisebox{-4.8em}{\quad=\quad}
	\Qcircuit @C=1em @R=.7em @!R {
		& \qw & \ctrl{3} & \qw & \ctrl{3} & \qw & \ctrl{3} & \qw & \qw \\
		& \vdots & & & & & & \vdots& \\
		& \qw & \ctrl{2} & \qw & \ctrl{2} & \qw & \ctrl{2} & \qw & \qw \\
		& \qw & \qw  & \ctrl{1} & \qw & \ctrl{1} & \ctrl{1} & \qw & \qw \\
		& \qw & \gate{(\omega S H )^\dagger} & \gate{S^\dagger} & \gate{\omega S H} & \gate{S} & \gate{iZ} & \qw & \qw
	}
	\]
	
	\[
	\Qcircuit @C=1em @R=1.1em @!R {
		& \qw & \ctrl{3} & \qw & \qw \\
		& \vdots & & \vdots& \\
		& \qw & \ctrl{2} & \qw & \qw  \\
		& \qw & \ctrl{1} & \qw & \qw \\
		& \qw & \gate{\omega H S} & \qw & \qw
	}
	\raisebox{-4.8em}{\quad=\quad}
	\Qcircuit @C=1em @R=.7em @!R {
		& \qw & \qw & \ctrl{3} & \qw & \ctrl{3} & \qw & \ctrl{3} & \qw & \qw & \qw \\
		& \vdots & & & & & & & & \vdots& \\
		& \qw & \qw & \ctrl{2} & \qw & \ctrl{2} & \qw & \ctrl{2} & \qw & \qw & \qw \\
		& \qw & \qw & \qw  & \ctrl{1} & \qw & \ctrl{1} & \ctrl{1} & \qw & \qw & \qw \\
		& \qw & \gate{\omega H} & \gate{(\omega S H )^\dagger} & \gate{S^\dagger} & \gate{\omega S H} & \gate{S} & \gate{iZ} & \gate{(\omega H)^\dagger} & \qw & \qw
	}
	\]
\end{proof}

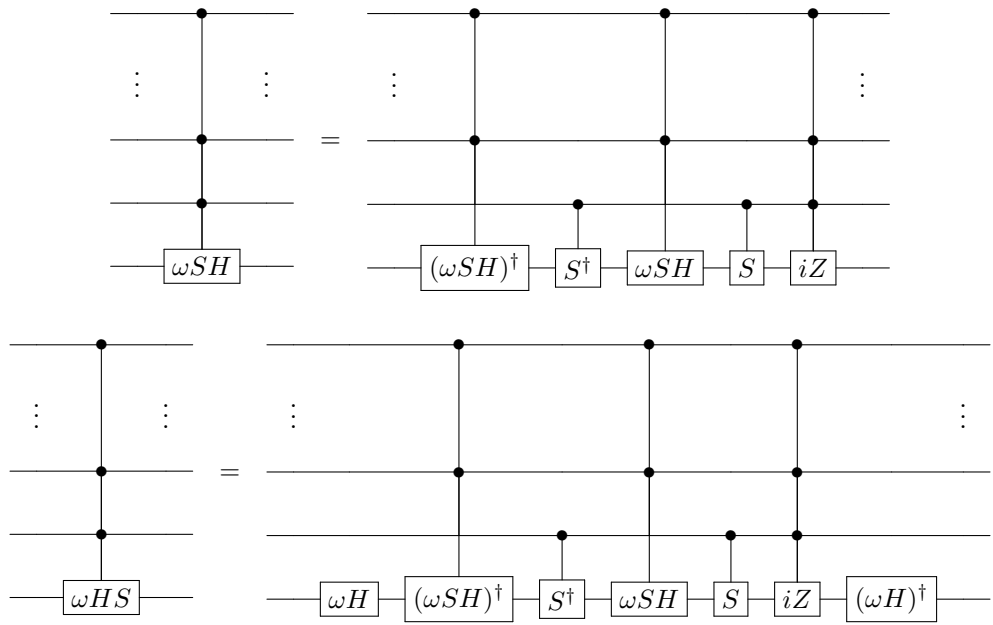
\end{document}